\documentclass[preprint,authoryear,11pt]{elsarticle}
\makeatletter
\def\ps@pprintTitle{%
 \let\@oddhead\@empty
 \let\@evenhead\@empty
 \def\@oddfoot{\centerline{\it{Preprint submitted to European Journal of Operational Research,  15 August 2018}}}%
 \let\@evenfoot\@oddfoot}
\makeatother

\usepackage{amsfonts,amsmath,amssymb,amsthm,url}
\usepackage{ulem}
\usepackage{url}
\usepackage{natbib}

\newtheorem{remark}{Remark}
\newtheorem{thm}{Theorem}
\newtheorem{defi}{Definition}
\newtheorem{props}{Proposition}
\usepackage{bm}
{\theoremstyle{THkey}}
\usepackage{algorithm}
\usepackage{algpseudocode}

\usepackage{subcaption}
\usepackage{multirow}
\usepackage{graphicx}
\usepackage{color}

\journal{European Journal of Operational Research}

\begin{document}

\begin{frontmatter}

\title{The impact of model risk on dynamic portfolio selection under multi-period mean-standard-deviation criterion}

\author[unsw]{Spiridon Penev\corref{cor1}}\ead{S.Penev@unsw.edu.au}
\author[mq]{Pavel V. Shevchenko}\ead{pavel.shevchenko@mq.edu.au}
\author[unsw]{Wei Wu}\ead{wei.wu1@unsw.edu.au}

\cortext[cor1]{Corresponding author}
\address[unsw]{School of Mathematics and Statistics, UNSW Sydney, NSW 2052 Australia}
\address[mq]{Department of Applied Finance and Actuarial Studies, Macquarie University, Sydney, NSW 2109 Australia}

\begin{abstract}
We quantify model risk of a financial portfolio whereby a multi-period mean-standard-deviation criterion is used as a selection criterion. In this work, model
risk is defined as the loss due to uncertainty of the underlying distribution of the returns of the assets in the portfolio. The uncertainty is measured by
the Kullback-Leibler divergence, i.e., the relative entropy. In the worst case scenario, the optimal robust strategy can be obtained in a semi-analytical form as
a solution of a system of nonlinear equations. Several numerical results are presented which allow us to compare the performance of this robust strategy with
the optimal non-robust strategy. For illustration, we also quantify the model risk associated with an empirical dataset.
\end{abstract}

\begin{keyword}
Multivariate statistics, Uncertainty modelling, Robust portfolio allocation, Pseudo dynamic programming, Mean-standard-deviation, Kullback-Leibler divergence
\end{keyword}
\end{frontmatter}

\section{Introduction.} \label{intro}
Portfolio selection has been studied extensively over the last few decades (see e.g., \cite{Mar52, LiN00}). Investors face the problem of choosing
the best possible investment strategy among thousands of assets. One significant difficulty in choosing optimal strategies is magnified by the fact that the
essential information needed to make an optimal decision, namely, the distribution of assets, is typically unknown or only vaguely known. Another difficulty is
that the distribution changes over time and a dynamic approach is needed to model it. A typical case  is when  the
cross-sectional distribution of the assets in the portfolio is assumed to be ``slightly deviating" from some nominal multivariate  distribution. The deviation
can be measured by a divergence measure such as the Kullback-Leibler (KL) divergence, i.e., the relative entropy, or more generally by the $\alpha$-divergence.
Intuitively, the bigger the divergence, the more significant  the impact on the optimal investment decision that is calculated under the nominal  distribution
assumption. However, the magnitude of the divergence that significantly impacts  the investment decision depends on the nominal distribution and on the type of
deviation. When the distributional assumptions are violated but only ``slightly", it may be prudent to use the optimal investment strategy under the nominal
distribution. This may in fact deliver better results (for example, in the sense of a higher expectation of the terminal wealth) since the robust approach--focusing on safeguarding against the worst possible outcome--may deliver too pessimistic a strategy that  may be disadvantageous when the nominal distribution
is violated only slightly. Ideally,  a ball of certain radius $\eta_0$ around the nominal distribution  is  given such that for
 distributions within the  radius the nominal distribution  is recommended, whereas when the model uncertainty
is bigger than $\eta_0$ then the robust approach is recommended. It is worth noting that although the optimal investment strategy under the nominal
distribution may perform better in the sense of a higher expectation of the terminal wealth,  the model risk as defined from the standard risk management
perspective may also be large. Thus, it is essential to quantify such model risk. \\

Unlike in robust optimization, where uncertainty is often measured by an uncertainty set (see e.g., \cite{KapCR14,KimKKF14}), the deviation between
distributions from a statistical point of view has been most commonly measured by a divergence measure. \cite{GlaX14} interpreted the KL divergence as a measure
of the amount of extra information required to adopt an alternative distribution, and disregard the nominal distribution. However, it has also been pointed out in
\cite{GlaX14} that the KL divergence is not suitable for heavy tailed distributions since  it relies on the assumption that the moment
generating function of the underlying random variables must exist in some open set containing the origin. Thus, the so-called $\alpha$-divergence is used as a
substitute. In contrast, \cite{SchS15} argued that the use of $\alpha$-divergence implicitly assumes that the tail of the deviating model is not heavier than that
of the nominal distribution. In fact, the popularity of the KL divergence and of the $\alpha$-divergence is due to the existence of a closed form solution when one
considers a worst case scenario approach or, in other words, the robust optimization approach to quantify model risk in risk management (see \cite{GlaX13},
\cite{SchS15}). If one only considers the measuring of model risk, alternative divergence measures are also possible. This has been discussed in the recent work of
\cite{BreC16} and in the references therein. Another interesting result worth mentioning in this area is the recent work by \cite{Lam16}, deriving an
asymptotic expansion of the worst  risk measure in the case of KL divergence. \\

The main  focus  in this paper is to investigate the impact of uncertainty of the distribution of returns of assets on the optimal portfolio allocation
model in \cite{BanGPW16}. As in \cite{BanGPW16}, the  selection criterion is the multi-period mean-standard-deviation and portfolio selection is performed in a dynamic
way. To measure uncertainty we use the KL divergence, which  is  reasonable if the underlying random variable (a function of asset returns) is not heavy tailed.
If we consider short term re-balancing (daily or weekly) as we do in our numerical examples, this is a reasonable and acceptable assumption.  Inspired
by \cite{KanF06}, we  find what we call a time consistent optimal robust strategy (see Definition \ref{defi:timeconsistency}). This reduces to solving a sequence
of single period portfolio selection problems. For each single period, we apply a robust optimization approach. Thus, we have to solve an inner and  an outer optimization
problem. The inner  problem is an infinite dimensional optimization  where we try to find a worst case distribution from a set of alternative distributions
(which have positive distances to the nominal distribution). A closed form solution to the inner optimization problem is available from past literature, see for example
\cite{Lam16}. The outer optimization problem is a standard convex optimization problem. By solving this, we derive a system of equations  which an optimal robust
strategy should satisfy. This is our first contribution. To be more precise, we have derived an optimal strategy in a semi-analytical form for the portfolio
selection model in \cite{BanGPW16} but with added  uncertainty of the distribution of the returns, where the uncertainty is measured by the KL divergence.
This optimal robust strategy can easily be calculated numerically in combination
with a simple Monte Carlo approach from \cite{GlaX14}. Our second contribution is to examine the impact of the uncertainty on portfolio selection by using the constructed
model.  Additionally, we compare  the performance of the optimal robust strategy and of the non-robust strategy (the optimal strategy when there is no distributional uncertainty)
under various scenarios. Moreover, we define model risk from the standard risk management perspective and quantify model risk using an empirical dataset. This provides a way
to examine model risk for practical purposes and is yet another contribution of our work.  \\

The study of the impact of uncertainty of the underlying distribution of  asset returns on the optimal strategy is an important one. In fact, the impact on the optimal
strategy ``under the worst case" was also raised as one of the five questions in the implementation of a robust risk management process by \cite{SchS15}. Their work focused
on the remaining four questions. Although there are several works devoted to the topic of  this   paper, there are some essential differences  to our work. \cite{Cal07} designed algorithms to solve mean-variance and mean-absolute-deviation static portfolio allocation under uncertainty. In contrast,
\cite{GlaX14} derived an analytical (or semi-analytical) solution for a static portfolio allocation problem under model uncertainty in which the mean-variance selection criterion
is used.
\cite{GlaX13} also explored a dynamic setting using a factor model.  However, their paper explicitly exploits an assumption of multivariate normality for the return vector. In our paper, no such assumption is required.
  \\

The paper is organized as follows. In Section 2, we define and formulate the problem of interest. In Section 3, we obtain the optimal robust strategy in a
semi-analytical form. Some discussions regarding the model and its computation are presented in Section 4. Section 5 is devoted  to   numerical examples and
discussions about quantifying model risk. We conclude the paper in Section 6.

\section{Problem Formulation.} \label{sec:probform}
We consider a market of $d > 1$ risky assets in which a risk free asset is not available. Suppose that an investor wants to invest all of  their money over
a fixed investment horizon $[0,N]$ into these $d$ risky assets. The return of each asset over the $n$th period $[n, n+1]$, $n = 0, ..., N-1$, is denoted
as $\bm{r}_{n+1} = (r^{1}_{n+1}, ..., r^{d}_{n+1})^\intercal$, where $r^{i}_{n+1}$, $i = 1, ...,d$  represents  the return of the $i$th asset over the $n$th
period. We assume that all random quantities are defined on a filtered complete probability space ($\Omega, \mathcal{F}, (\mathcal{F}_{n}), \mathbb{P}$)
with the sample space $\Omega$, the sigma-algebra $\mathcal{F}$, filtration $(\mathcal{F}_{n})$, and the probability measure $\mathbb{P}$, where the
sigma-algebra $\mathcal{F}_{n} = \sigma(\bm{r}_{m}, 1 \leq m \leq n )$ and $\mathcal{F}_{0}$ is trivial. Moreover, the return vector $\bm{r}_{n+1}$ has
finite second moments in the $L^{2}$ sense, i.e., $\mathbb{E}(|\bm{r}_{n+1}|^2) < \infty$.  \\

At each time $n = 0, ..., N-1$, the investor re-balances the portfolio using a re-balancing strategy $\bm{u}= (\bm{u}_{0}, ... ,\bm{u}_{N-1})^\intercal$,
where $u_{n}^{i} \in \mathbb{R}$, $i = 1, ..., d$ denotes the proportional allocation of the wealth of the investor into the $i$th asset. We
denote by $\mathcal U^{0}$ the set of admissible strategies at time $0$ such that all $\bm{u}_{n}$, where $n = 0, ..., N-1$, take values in the set
\small
\begin{eqnarray}
  U = \Big\{\bm{u} \in \mathbb{R}^{d}: \bm{1}^\intercal\bm{u} = 1\Big\}. \nonumber
\end{eqnarray}
\normalsize
For $m > 0$, we use $\mathcal{U}^{m}$ to denote the set of admissible sub-strategies $\bm{u}^{m} = (\bm{u}_{n})_{n \geq m}$ such that all $\bm{u}_{n}, n = m,
..., N-1$, take values in the set $U$. \\

Let $W_{n}$ denote the wealth of the investor at time $n$, where $n = 0, ..., N$. We assume that $W_{n}$ and $\bm{r}_{n+1}$ are independent. This can be achieved,
for example, by assuming that $\bm{r}_{n}$ are independent, identically distributed. During the period $[n, n+1]$, the investor's wealth changes to
\small
\begin{eqnarray}
  W_{n+1} &   =   &  W_{n}(\bold{1}+\bm{r}_{n+1})^\intercal\bm{u}_{n} = W_{n}\bm{R}_{n+1}^\intercal\bm{u}_{n},     \nonumber
\end{eqnarray}
\normalsize
where $\bm{R}_{n+1} = \bold{1}+\bm{r}_{n+1}$. For $x \in \mathbb{R}$, at any time $m$, the aim of the investor is to optimize
\small
\begin{eqnarray}\label{eqn:mstdobj}
    J_{m,x}(\bm{u}^{m}) &  =  &    \mathbb{E}\Big(\sum_{n=m}^{N-1}\mathcal{J}_{n,W_{n}}(W_{n+1}) | W_{m} = x\Big),
\end{eqnarray}
\normalsize
where
\small
\begin{eqnarray}
  \mathcal{J}_{n,W_{n}}(W_{n+1})  &   =   &    W_{n+1} - \kappa_{n}\sqrt{Var_{n,W_{n}}(W_{n+1})}     \nonumber \\
                                  &   =   &    W_{n}\Big(\bm{R}_{n+1}^\intercal\bm{u}_{n}  - \kappa_{n}\sqrt{\bm{u}^\intercal_{n}\bm{\Sigma}_{n}\bm{u}_{n}}\Big), \nonumber \\
        Var_{n,W_{n}}(W_{n+1})    &   =   &    Var(W_{n+1} | W_{n}),    \ \ \  \textnormal{and}    \ \ \
          \bm{\Sigma}_{n}             =        Var(\bm{r}_{n+1}), \nonumber
\end{eqnarray}
\normalsize
where $Var(\bm{r}_{n+1})$ is the variance of $\bm{r}_{n+1}$, and the parameter $\kappa_{n}$ characterizes the risk aversion of the investor. The above criterion is a
multi-period selection criterion of mean-standard-deviation (MSD) type. We note that because of the scaling property of the single period mean-standard-deviation criterion,
the optimization of the intermediate wealth contributes directly to the optimization of the terminal wealth. For more properties and discussions of this objective we refer
to \cite{BanGPW16}.

\begin{remark}
  One may note that we do not include any discounting here to reflect the time value of money. This is because we assume a market of risky assets only. Since no risk free
  asset is available, disregarding discounting would be appropriate.
\end{remark}

The value function of this control problem takes the form
\small
\begin{eqnarray}\label{eqn:nonrobust}
  \mathcal{V}(m,x) = \sup_{\bm{u}^{m} \in \mathcal{U}^{m}}J_{m,x}(\bm{u}^{m}).
\end{eqnarray}
\normalsize

We will use the Kullback-Leibler (KL) divergence

\small
\begin{eqnarray*}
 \mathcal{R}(\mathcal{E}) &   =   &   \mathbb{E}\Big(\mathcal{E}\log\mathcal{E}\Big),
\end{eqnarray*}
\normalsize
where $\mathcal{E}$ is the ratio of the density of an alternative distribution to the density of the nominal distribution, as a deviation measure between different
distributions. For the reader's convenience we recall the concept of the KL divergence in Appendix A. Now, for a given $\eta > 0$, a KL
divergence ball is defined as
\small
\begin{eqnarray}\label{eqn:klball}
  \mathcal{B}_{\eta} = \{\mathcal{E}: \mathcal{R}(\mathcal{E}) \leq \eta\}.
\end{eqnarray}
\normalsize

Next, to quantify the model risk for the investor, we formulate a robust version of the problem in (\ref{eqn:nonrobust}). We first define a sequence of KL divergence
balls
\small
\begin{eqnarray}
  \mathcal{B}_{\eta_{n}} = \{\mathcal{E}: \mathcal{R}(\mathcal{E}) \leq \eta_{n}\}, \ \ \textnormal{where} \ n = 0, ..., N-1.     \nonumber
\end{eqnarray}
\normalsize
One may note that all moments through this paper are defined with respect to the nominal distribution. \\

Given any starting time $m = 0, ..., N-1$, we denote the set of $\bm{\mathcal{E}}^{m} = (\mathcal{E}_{m}, ..., \mathcal{E}_{N-1})$ such that each $\mathcal{E}_{n} \in
\mathcal{B}_{\eta_{n}}$, where $n = m, ..., N-1$, by $\mathcal{B}^{m}$. Here, $\mathcal{E}_{n}$, $n = m, ..., N-1$, is the ratio of the density of an alternative
distribution to the density of the nominal distribution over $[n, n+1]$. The robust version of (\ref{eqn:nonrobust}) is defined then by
\small
\begin{eqnarray}\label{eqn:robust}
   V(m,x)    &  = &    \displaystyle\sup_{\bm{u}^{m} \in \mathcal{U}^{m}}\inf_{\bm{\mathcal{E}}^{m} \in \mathcal{B}^{m}} J_{m,x}(\bm{\mathcal{E}}^{m},\bm{u}^{m}),
\end{eqnarray}
\normalsize
where
\small
\begin{eqnarray}
  J_{m,x}(\bm{\mathcal{E}}^{m},\bm{u}^{m})  &  =  &    \mathbb{E}\Bigg(\mathcal{E}_{m}W_{m}\Big(\bm{R}_{m+1}^\intercal\bm{u}_{m}
                                                     - \kappa_{m}\sqrt{\bm{u}^\intercal_{m}\bm{\Sigma}_{m}\bm{u}_{m}}\Big)       \nonumber \\
                                            &     &  + \sum_{n=m+1}^{N-1}e^{-\eta_{n}c_{n}\kappa_{n}}\mathcal{E}_{n}\mathcal{J}_{n,W_{n}}(W_{n+1}) | W_{m} = x\Bigg),  \nonumber
\end{eqnarray}
\normalsize
where $c_{n}$ are scaling parameters. Note that  if $m \geq N-1$, the summation term is set to zero.

\begin{remark}
  Here we take the infimum over the set of all possible distributions within a KL ball. This corresponds to the worst case scenario. We then find the best strategy
  under the worst case scenario. The robustification process has an interpretation from game theory. For readers interested in finding more about this interpretation,
  we refer to \cite{GlaX14}.
\end{remark}

\begin{remark}
  At time $n$, we scale the future payoff by a factor $e^{-\eta_{n}c_{n}\kappa_{n}}$. The term is added for mathematical convenience and guarantees the existence of
  an optimal solution. Indeed, we reduce the  impact of the future uncertainty on the current stage according to the investor's risk aversion and the radius of the
  divergence ball in the following period.
  It is worth noting that as the radius of
  the divergence ball $\eta_{n} \rightarrow 0$, we return to the non-robust case.
\end{remark}

\section{Semi-Analytical Optimal Solution under KL Divergence.} \label{sec:semianalyopsol}
To solve the robust control problem in (\ref{eqn:robust}), we  apply a strategy  that we call a strongly time consistent optimal robust strategy. It represents a
robustified version of a strong time consistent optimal strategy  inspired by \cite{KanF06} (see also \cite[definition 2]{BanGPW16}). The exact definition is
given below.

\begin{defi}\label{defi:timeconsistency}
   Given any starting time $m = 0, ..., N-1$, a strategy $\bm{u}^{m,\ast} = (\bm{u}^{\ast}_{m}, ..., \bm{u}^{\ast}_{N-1})$ is said
   to be a strongly time consistent optimal robust strategy with respect to $J_{m,x}(\bm{\mathcal{E}}^{m},\bm{u}^{m})$ if it satisfies the
   following two conditions.
   \begin{itemize}
     \item \textbf{Condition 1}: Let $\mathcal{A}^{m} \subset \mathcal{U}^{m}$ be a set of strategies of the form $\bm{u}^{m} = (\bm{v},
                                 \bm{u}^{\ast}_{m+1}, ..., \bm{u}^{\ast}_{N-1})$, where $\bm{v} \in \mathbb{R}^{d}$ is arbitrary.
                                 Then there exists $\bm{\mathcal{E}}^{m,\ast}(\cdot) = (\bm{\mathcal{E}}_{m}^{\ast}(\cdot),
                                 \bm{\mathcal{E}}_{m}^{\ast}(\cdot), ..., \bm{\mathcal{E}}_{N-1}^{\ast}(\cdot)) \in \mathcal{B}^{m}$
                                 such that
                                 \small
                                 \begin{eqnarray}
                                             \inf_{\bm{\mathcal{E}}^{m} \in \mathcal{B}^{m}} J_{m,x}(\bm{\mathcal{E}}^{m},\cdot)
                                    &  =  &  J_{m,x}(\bm{\mathcal{E}}^{m,\ast}(\cdot),\cdot). \\
                                             \sup_{\bm{u}^{m} \in \mathcal{A}^{m}}J_{m,x}(\bm{\mathcal{E}}^{m,\ast}(\bm{u}^{m}),\bm{u}^{m})
                                    &  =  &  J_{m,x}(\bm{\mathcal{E}}^{m,\ast}(\bm{u}^{m,\ast}),\bm{u}^{m,\ast}).
                                 \end{eqnarray}
                                 \normalsize
     \item \textbf{Condition 2}: For $n = m+1, ..., N-1$, $\bm{u}^{n} = (\bm{u}_{n}, \bm{u}_{n+1}, ..., \bm{u}_{N-1}) \in \mathcal{U}^{n}$,
                                 there exists $\bm{\mathcal{E}}^{n,\ast}(\cdot) = (\bm{\mathcal{E}}_{n}^{\ast}(\cdot),
                                 \bm{\mathcal{E}}_{n+1}^{\ast}(\cdot) ..., \bm{\mathcal{E}}_{N-1}^{\ast}(\cdot))$ \mbox{$\in \mathcal{B}^{n}$}
                                 such that \small
                                 \begin{eqnarray}
                                            \inf_{\bm{\mathcal{E}}^{n} \in \mathcal{B}^{n}} J_{n,x}(\bm{\mathcal{E}_{n}},\cdot)
                                   &  =  &  J_{n,x}(\bm{\mathcal{E}}^{n,\ast}(\cdot), \cdot).   \\
                                            \sup_{\bm{u}^{n} \in \mathcal{U}^{n}}J_{n,x}(\bm{\mathcal{E}}^{n,\ast}(\bm{u}^{n}), \bm{u}^{n})
                                   &  =  &  J_{n,x}(\bm{\mathcal{E}}^{n,\ast}(\bm{u}^{n,\ast}), \bm{u}^{n,\ast}).
                                 \end{eqnarray}
                                 \normalsize
   \end{itemize}
   If only \textbf{Condition 1} is satisfied then we say that the strategy is a weakly time consistent optimal robust strategy with respect to $J_{m,x}(\cdot)$.
\end{defi}

\begin{remark}
 The time consistency that we defined here refers to the time consistency of a strategy with respect to the particular criterion that we choose. There are other
 definitions of time consistency such as the time consistency of a selection criterion itself (see \cite[definition 2]{CheLG13}).
\end{remark}

Since the value function of the robust control problem given by (\ref{eqn:robust}) is separable (in the sense that it can be written as a sum of expectations), by a
similar argument as in the proof of \cite[theorem 3]{CheLG13} we know that a weakly time consistent optimal strategy, which can be found by period-wise
optimization, is also a strongly time consistent optimal strategy. \\

The rest of this section is devoted to the following theorem and its proof. This theorem summarizes one of our main findings, i.e., a system of nonlinear equations that
an optimal strategy  should  satisfy.

\begin{thm}\label{thm:opstrategysys}
  Suppose that $(\bm{u}^{\ast}_{m})$, $m = 0, 1, ..., N-1$ is a strategy where there exists a sequence $(\theta_{m}^{\ast})$, with $\theta_{m}^{\ast} > 0$, such that
  \small
  \begin{eqnarray}
     \mathbb{E}\Big(\exp\big(-\bm{R}^\intercal_{m+1}\bm{u}^{\ast}_{m}\frac{2}{\theta_{m}^{\ast}}\big)\Big) < \infty,  \nonumber
  \end{eqnarray}
  \normalsize
  and
  \small
  \begin{eqnarray}
   \bm{u}^{\ast}_{m}                   &   =   &    \frac{S_{m}^{\ast}}{\kappa_{m}}\Big(\bm{\Sigma}_{m}^{-1}\bm{X}^{\ast}_{m}
                                                  - \frac{b_{m}^{\ast}\bm{\Sigma}_{m}^{-1}\bm{1}}{a_{m}}\Big)
                                                  + \frac{\bm{\Sigma}_{m}^{-1}\bm{1}}{a_{m}}, \label{eqn:opsolsys1} \\
         S^{\ast}_{m}                  &   =   &    \sqrt{\frac{\frac{1}{a_{m}}}{1 - \frac{h_{m}^{\ast}}{\kappa_{m}^{2}}
                                                  + \frac{(b_{m}^{\ast})^2}{\kappa_{m}^{2} a_{m}}}}
                                           =        \sqrt{\frac{\frac{1}{a_{m}}}{1 - \frac{1}{\kappa_{m}^{2}}g_{m}^{\ast}}},  \label{eqn:opsolsys2} \\
      \bm{X}^{\ast}_{m}                &   =   &    \frac{\mathbb{E}\Big(\exp\big(-\bm{R}_{m+1}^\intercal\bm{u}_{m}^{\ast}\frac{1}{\theta_{m}^{\ast}}\big)\bm{R}_{m+1}\Big)}
                                                    {\mathbb{E}\Big(\exp\big(-\bm{R}_{m+1}^\intercal\bm{u}_{m}^{\ast}\frac{1}{\theta_{m}^{\ast}}\big)\Big)}
                                                    + e^{-\eta_{m+1}c_{m+1}\kappa_{m+1}} \times \nonumber \\
                                       &       &    G_{m+1}(\bm{u}_{m+1}^{\ast}, \theta_{m+1}^{\ast})\mathbb{E}(\bm{R}_{m+1}),
                                                    \ \ \ \  \label{eqn:opsolsys3} \\
  \mathbb{E}\big(\mathcal{E}_{m}^{\ast}\log(\mathcal{E}_{m}^{\ast})\big)   &    =   &   \eta_{m},  \label{eqn:opsolsys4}
 \end{eqnarray}
 \normalsize
 where
 \small
  \begin{eqnarray}
    &     &  g_{m}^{\ast} = h_{m}^{\ast} - \frac{(b^{\ast}_{m})^{2}}{a_{m}}, \ \ \ h_{m}^{\ast} = (\bm{X}_{m}^{\ast})^{T}\bm{\Sigma}_{m}^{-1}\bm{X}_{m}^{\ast},
             \ \ \ a_{m} = \bm{1}^\intercal\bm{\Sigma}_{m}^{-1}\bm{1}, \nonumber \\
    &     &  b_{m}^{\ast} = \bm{1}^\intercal\bm{\Sigma}_{m}^{-1}\bm{X}_{m}^{\ast}, \ \ \
             \mathcal{E}^{\ast}_{m} = \displaystyle\frac{\exp\big(-\bm{R}_{m+1}^\intercal\bm{u}_{m}^{\ast}\frac{1}{\theta_{m}^{\ast}}\big)}
             {\mathbb{E}\Big(\exp\big(-\bm{R}_{m+1}^\intercal\bm{u}_{m}^{\ast}\frac{1}{\theta_{m}^{\ast}}\big)\Big)} \ \ \ \mathbb{P}\textnormal{-a.s.},    \nonumber \\
   &     &   G_{m}(\bm{u}_{m}^{\ast}, \theta_{m}^{\ast})                \nonumber \\
   &  =  & - \theta^{\ast}_{m}\log\mathbb{E}\Big(\exp\big(-\bm{R}_{m+1}^\intercal\bm{u}_{m}^{\ast}\frac{1}{\theta_{m}^{\ast}}\big)\Big)
           + e^{-\eta_{m+1}c_{m+1}\kappa_{m+1}} \times \nonumber  \\
   &     &   G_{m+1}(\bm{u}_{m+1}^{\ast}, \theta_{m+1}^{\ast}) \mathbb{E}\big(\bm{R}_{m+1}^\intercal\bm{u}_{m}^{\ast}\big) - \kappa_{m}S_{m}^{\ast} - \eta_{m}\theta_{m}^{\ast},   \nonumber
 \end{eqnarray}
 \begin{eqnarray}
             G_{N}(\bm{u}_{N}^{\ast}, \theta_{N}^{\ast})
   &  =  &   0.  \nonumber
 \end{eqnarray}
 \normalsize
 Then, $(\bm{u}^{\ast}_{m})$  is optimal, and the value function is given by
 \small
 \begin{eqnarray}
    V(m,x) = xG_{m}(\bm{u}_{m}^{\ast}, \theta_{m}^{\ast}),   \nonumber
 \end{eqnarray}
 \normalsize
 where $x \in (0,\infty)$.
\end{thm}

\begin{proof}
 We proceed by using an induction argument. \\

 \textbf{Step 1}: Firstly, let $m = N-1.$ The optimization problem  becomes
 \small
 \begin{eqnarray}
   &    \displaystyle\sup_{\bm{u}_{N-1} \in U}\inf_{\mathcal{E}_{N-1} \in \mathcal{B}_{\eta_{N-1}}}     &
   \mathbb{E}\Bigg(\mathcal{E}_{N-1}W_{N-1}\Big(\bm{u}^\intercal_{N-1}\bm{R}_{N} \nonumber \\
   &                                                                                                    &
   - \kappa_{N-1}\sqrt{\bm{u}^\intercal_{N-1}\bm{\Sigma}_{N-1}\bm{u}_{N-1}}  \ \Big) \Big| W_{N-1} = x\Bigg).   \nonumber
 \end{eqnarray}
 \normalsize
 which reduces to
 \small
 \begin{eqnarray}
   &    \displaystyle\sup_{\bm{u}_{N-1} \in U}\inf_{\mathcal{E}_{N-1} \in \mathcal{B}_{\eta_{N-1}}}     &
   \Bigg(\mathbb{E}\big(\mathcal{E}_{N-1}\bm{u}^\intercal_{N-1}\bm{R}_{N}\big) - \kappa_{N-1}\sqrt{\bm{u}^\intercal_{N-1}\bm{\Sigma}_{N-1}\bm{u}_{N-1}} \ \Bigg). \nonumber \\
   \label{eqn:step1op}
 \end{eqnarray}
 \normalsize
 Let us look at the inner optimization problem, i.e.,
 \small
  \begin{eqnarray}\label{eqn:step1opinner}
   &   \displaystyle\inf_{\mathcal{E}_{N-1} \in \mathcal{B}_{\eta_{N-1}}}   &   \Bigg(\mathbb{E}\big(\mathcal{E}_{N-1}\bm{u}^\intercal_{N-1}\bm{R}_{N}\big)
                                                                              - \kappa_{N-1}\sqrt{\bm{u}^\intercal_{N-1}\bm{\Sigma}_{N-1}\bm{u}_{N-1}}\ \Bigg).
 \end{eqnarray}
 \normalsize
 We can write down the Lagrangian as
 \small
 \begin{eqnarray}
   L_{N-1}(\mathcal{E}_{N-1},\theta_{N-1}) &  =  &    \mathbb{E}\Big(\mathcal{E}_{N-1}\bm{u}^\intercal_{N-1}\bm{R}_{N}\Big)
                                                    - \kappa_{N-1}\sqrt{\bm{u}^\intercal_{N-1}\bm{\Sigma}_{N-1}\bm{u}_{N-1}}                \nonumber \\
                                           &     &   + \theta_{N-1}\Big(\mathbb{E}\big(\mathcal{E}_{N-1}\log(\mathcal{E}_{N-1})\big) - \eta_{N-1}\Big).  \nonumber
 \end{eqnarray}
 \normalsize
 By setting the derivative (with respect to $\mathcal{E}_{N-1}$) of the expression under the expectation of the Lagrangian to be equal to zero, we obtain
 \small
 \begin{eqnarray}
    \bm{u}^\intercal_{N-1}\bm{R}_{N} + \theta_{N-1}\log(\mathcal{E}_{N-1}) + \theta_{N-1}    &   =   &   0.  \ \ \ \mathbb{P}\textnormal{-a.s.}, \nonumber
 \end{eqnarray}
 \normalsize
 Solving the above equation together with the fact that all alternative distributions have a proper density, i.e.,
 \small
 \begin{eqnarray}
   \mathbb{E}(\mathcal{E}_{N-1})    &  =  &  1, \nonumber
 \end{eqnarray}
 \normalsize
 we obtain
 \small
 \begin{eqnarray}\label{eqn:solKLop}
   \mathcal{E}_{N-1}^{\ast}   &   =   &   \displaystyle\frac{\exp\big(-\bm{R}_{N}^\intercal\bm{u}_{N-1}\frac{1}{\theta_{N-1}}\big)}
                                          {\mathbb{E}\Big(\exp\big(-\bm{R}_{N}^\intercal\bm{u}_{N-1}\frac{1}{\theta_{N-1}}\big)\Big)}
                                          \ \ \ \mathbb{P}\textnormal{-a.s.},
 \end{eqnarray}
 \normalsize
 for some $\theta_{N-1}$ such that
 \small
 \begin{eqnarray}
    \mathbb{E}\Big(\exp\big(-\bm{R}^\intercal_{N}\bm{u}^{\ast}_{N-1}\frac{1}{\theta_{N-1}^{\ast}}\big)\Big) < \infty.  \nonumber
 \end{eqnarray}
 \normalsize
 We can verify that (\ref{eqn:solKLop}) is indeed the optimal solution by using a convexity argument. We refer to the proof of \cite[proposition 3.1]{Lam16}
 for more details. \\

 Now, since the set
 \small
 \begin{eqnarray}
   \{\mathcal{E}: \mathcal{R}(\mathcal{E}) < \eta_{N-1}\} \nonumber
 \end{eqnarray}
 \normalsize
 is not empty, by \cite[theorem 2.1]{BenTC88}, strong duality holds. This implies (see \cite[pp. 242--243]{BoyV09}) that  the optimal solution
 $\mathcal{E}^{\ast}_{N-1}$ and its corresponding $\theta_{N-1}$ satisfies the following system:
 \small
 \begin{eqnarray}
    \theta_{N-1}\Big(\mathbb{E}\big(\mathcal{E}^{\ast}_{N-1}\log(\mathcal{E}^{\ast}_{N-1})\big) - \eta_{N-1}\Big)   &   =   &  0, \nonumber \\
    \mathbb{E}\big(\mathcal{E}^{\ast}_{N-1}\log(\mathcal{E}^{\ast}_{N-1})\big)    &  \leq  &  \eta_{N-1}, \nonumber \\
                                                      \theta_{N-1}  &   >    &  0.    \nonumber
 \end{eqnarray}
 \normalsize
 We denote the solution $\theta_{N-1}$ of this system as $\theta_{N-1}^{\ast}$. \\

 Next, with (\ref{eqn:solKLop}), the optimization problem (\ref{eqn:step1op}) becomes
 \small
 \begin{eqnarray}
   &    \displaystyle\sup_{\bm{u}_{N-1} \in U}   &    \Bigg(-\theta_{N-1}^{\ast}\log\mathbb{E}\Big(\exp\big(-\bm{R}_{N}^\intercal\bm{u}_{N-1}
                                                      \frac{1}{\theta_{N-1}^{\ast}}\big)\Big)    \nonumber \\
   &                                             &  - \kappa_{N-1}\sqrt{\bm{u}^\intercal_{N-1}\bm{\Sigma}_{N-1}\bm{u}_{N-1}}
                                                    - \eta_{N-1}\theta_{N-1}^{\ast}\Bigg).      \label{eqn:opuN-1}
 \end{eqnarray}
 \normalsize
 One may note that the expression under the supremum is actually the optimal dual of (\ref{eqn:step1opinner}). This can be confirmed by applying \cite[lemma 2.1]{BenTC88}. \\

 By Lemma B.1 (see Appendix B), we obtain the unique optimum of (\ref{eqn:opuN-1}) which satisfies the following system of nonlinear equations:
 \small
 \begin{eqnarray}
   \bm{u}^{\ast}_{N-1}    &   =   &   \frac{S_{N-1}^{\ast}}{\kappa_{N-1}}\Big(\bm{\Sigma}_{N-1}^{-1}\bm{X}^{\ast}_{N-1}
                                    - \frac{b_{N-1}^{\ast}\bm{\Sigma}_{N-1}^{-1}\bm{1}}{a_{N-1}}\Big) + \frac{\bm{\Sigma}_{N-1}^{-1}\bm{1}}{a_{N-1}}, \nonumber \\
         S^{\ast}_{N-1}   &   =   &  \sqrt{\frac{\frac{1}{a_{N-1}}}{1 - \frac{h_{N-1}}{\kappa_{N-1}^2} + \frac{(b_{N-1}^{\ast})^2}{\kappa_{N-1}^2 a_{N-1}}}}
                              =      \sqrt{\frac{\frac{1}{a_{N-1}}}{1 - \frac{1}{\kappa_{N-1}^2}g_{N-1}^{\ast}}}, \nonumber
 \end{eqnarray}
 \begin{eqnarray}
      \bm{X}^{\ast}_{N-1} &   =   &  \frac{\mathbb{E}\Big(\exp(-\bm{R}^\intercal_{N}\bm{u}_{N-1}^{\ast}\frac{1}{\theta_{N-1}^{\ast}})\bm{R}_{N}\Big)}
                                     {\mathbb{E}\Big(\exp(-\bm{R}_{N}^\intercal\bm{u}_{N-1}^{\ast}\frac{1}{\theta_{N-1}^{\ast}})\Big)}. \nonumber
 \end{eqnarray}
 \normalsize
 This follows from the proof of \cite[theorem 4]{BanGPW16}. Thus, the corresponding value function is given by
 \small
 \begin{eqnarray}
   V(N-1,x) = xG_{N-1}(\bm{u}_{N-1}^{\ast}, \theta_{N-1}^{\ast}),   \nonumber
 \end{eqnarray}
 \normalsize
 where
 \small
 \begin{eqnarray}
             G_{N-1}(\bm{u}_{N-1}^{\ast}, \theta_{N-1}^{\ast})
   &  =  & - \theta_{N-1}^{\ast}\log\mathbb{E}\Big(\exp\big(-\bm{R}_{N}^\intercal\bm{u}_{N-1}^{\ast}\frac{1}{\theta_{N-1}^{\ast}}\big)\Big) \nonumber \\
   &     & - \kappa_{N-1}S_{N-1}^{\ast} - \eta_{N-1}\theta_{N-1}^{\ast}.   \nonumber
 \end{eqnarray}
 \normalsize
 \textbf{Step 2}: Now, when $m = N-2$, the optimization problem becomes
 \small
 \begin{eqnarray}\label{eqn:step2op}
   &    \displaystyle\sup_{\bm{u}_{N-2} \in U}\inf_{\mathcal{E}_{N-2} \in \mathcal{B}_{\eta_{N-2}}}     &
        \Bigg(\mathbb{E}\Big(\mathcal{E}_{N-2}\bm{R}_{N-1}^\intercal\bm{u}_{N-2}\Big)
      + e^{-\eta_{N-1}c_{N-1}\kappa_{N-1}}\times \nonumber \\
   &                                                                                                    &
        G_{N-1}(\bm{u}_{N-1}^{\ast}, \theta_{N-1}^{\ast})\mathbb{E}\Big(\bm{R}_{N-1}^\intercal\bm{u}_{N-2}\Big) \nonumber \\
  &                                                                                                     &
      - \kappa_{N-2}\sqrt{\bm{u}^\intercal_{N-2}\bm{\Sigma}_{N-2}\bm{u}_{N-2}}\Bigg).
 \end{eqnarray}
 \normalsize
 Again, let us write down the Lagrangian
 \small
 \begin{eqnarray}
               L_{N-2}(\mathcal{E}_{N-2},\theta_{N-2})
   &  =  &     \mathbb{E}\Big(\mathcal{E}_{N-2}\bm{R}_{N-1}^\intercal\bm{u}_{N-2}\Big)
            +  e^{-\eta_{N-1}c_{N-1}\kappa_{N-1}}\times    \nonumber  \\
   &     &     G_{N-1}(\bm{u}_{N-1}^{\ast}, \theta_{N-1}^{\ast})\mathbb{E}\Big(\bm{R}_{N-1}^\intercal\bm{u}_{N-2}\Big)        \nonumber \\
   &     &  -  \kappa_{N-2}\sqrt{\bm{u}^\intercal_{N-2}\bm{\Sigma}_{N-2}\bm{u}_{N-2}}                            \nonumber  \\
   &     &  +  \theta_{N-2}\Big(\mathbb{E}\big(\mathcal{E}_{N-2}\log(\mathcal{E}_{N-2})\big) - \eta_{N-2}\Big).  \nonumber
 \end{eqnarray}
 \normalsize
 As in \textbf{Step 1}, we obtain the optimal $\mathcal{E}_{N-2}^{\ast}$:
 \small
 \begin{eqnarray}\label{eqn:KLop2}
   \mathcal{E}_{N-2}^{\ast}   &   =   &   \displaystyle\frac{\exp\big(-\bm{R}_{N-1}^\intercal\bm{u}_{N-2}\frac{1}{\theta_{N-2}}\big)}
                                         {\mathbb{E}\Big(\exp\big(-\bm{R}_{N-1}^\intercal\bm{u}_{N-2}\frac{1}{\theta_{N-2}}\big)\Big)} \ \ \ \mathbb{P}\textnormal{-a.s.},
 \end{eqnarray}
 \normalsize
 by solving
 \small
 \begin{eqnarray}
    \bm{u}^\intercal_{N-2}\bm{R}_{N-1} + \theta_{N-2}\log(\mathcal{E}_{N-2}) + \theta_{N-2}    &   =   &   0, \ \ \ \mathbb{P}\textnormal{-a.s.}, \nonumber
 \end{eqnarray}
 \normalsize
 together with the fact that all alternative distributions have a proper density, i.e.,
 \small
 \begin{eqnarray}
   \mathbb{E}(\mathcal{E}_{N-2})    &  =  &  1. \nonumber
 \end{eqnarray}
 \normalsize
  Moreover, the optimal $\mathcal{E}_{N-2}^{\ast}$ and its associated optimal $\theta_{N-2}$ satisfy the following system:
  \small
 \begin{eqnarray}
    \theta_{N-2}\Big(\mathbb{E}\big(\mathcal{E}_{N-2}\log(\mathcal{E}_{N-2})\big) - \eta_{N-2}\Big)   &   =   &  0, \nonumber \\
    \mathbb{E}\big(\mathcal{E}_{N-2}\log(\mathcal{E}_{N-2})\big)    &  \leq  &  \eta_{N-2}, \nonumber \\
                                                   \theta_{N-2}     &   >    &  0.    \nonumber
 \end{eqnarray}
 \normalsize
 We denote such $\theta_{N-2}$ as $\theta_{N-2}^{\ast}$. \\

 Now, with (\ref{eqn:KLop2}), the optimization problem (\ref{eqn:step2op}) becomes
 \small
 \begin{eqnarray}\label{eqn:opuN-2}
    &   \displaystyle\sup_{\bm{u}_{N-2} \in U}  &
       \Bigg(-\theta_{N-2}^{\ast}\log\mathbb{E}\Big(\exp\big(-\bm{R}_{N-1}^\intercal\bm{u}_{N-2}\frac{1}{\theta_{N-2}^{\ast}}\big)\Big)
    +  e^{-\eta_{N-1}c_{N-1}\kappa_{N-1}}\times \nonumber \\
    &                                           &  G_{N-1}(\bm{u}_{N-1}^{\ast}, \theta_{N-1}^{\ast})\mathbb{E}\Big(\bm{R}_{N-1}^\intercal\bm{u}_{N-2}\Big)         \nonumber \\
    &                                           &  - \kappa_{N-2}\sqrt{\bm{u}^\intercal_{N-2}\bm{\Sigma}_{N-2}\bm{u}_{N-2}} - \eta_{N-2}\theta_{N-2}^{\ast}\Bigg).
 \end{eqnarray}
 \normalsize

 By Lemma B.1, we obtain the unique optimum of (\ref{eqn:opuN-2}) which satisfies the following system of nonlinear equations:
 \small
 \begin{eqnarray}
   \bm{u}^{\ast}_{N-2}    &   =   &   \frac{S_{N-2}^{\ast}}{\kappa_{N-2}}\Big(\bm{\Sigma}_{N-2}^{-1}\bm{X}^{\ast}_{N-2}
                                    - \frac{b_{N-2}^{\ast}\bm{\Sigma}_{N-2}^{-1}\bm{1}}{a_{N-2}}\Big) + \frac{\bm{\Sigma}_{N-2}^{-1}\bm{1}}{a_{N-2}}, \nonumber \\
         S^{\ast}_{N-2}   &   =   &   \sqrt{\frac{\frac{1}{a_{N-2}}}{1 - \frac{h_{N-2}}{\kappa_{N-2}^2} + \frac{(b_{N-2}^{\ast})^2}{\kappa_{N-2}^2 a_{N-2}}}}
                                    = \sqrt{\frac{\frac{1}{a_{N-2}}}{1 - \frac{1}{\kappa_{N-2}^2}g_{N-2}^{\ast}}}, \nonumber \\
      \bm{X}^{\ast}_{N-2} &   =   &   \frac{\mathbb{E}\Big(\exp(-\bm{R}^\intercal_{N-1}\bm{u}_{N-2}^{\ast}\frac{1}{\theta_{N-2}^{\ast}})\bm{R}_{N-1}\Big)}
                                      {\mathbb{E}\Big(\exp(-\bm{R}_{N-1}^\intercal\bm{u}_{N-2}^{\ast}\frac{1}{\theta_{N-2}^{\ast}})\Big)}     \nonumber \\
                          &       & + e^{-\eta_{N-1}c_{N-1}\kappa_{N-1}}G_{N-1}(\bm{u}_{N-1}^{\ast}, \theta_{N-1}^{\ast})\mathbb{E}(\bm{R}_{N-1}). \nonumber
 \end{eqnarray}
 \normalsize
 This again follows from the proof of  \cite[theorem 4]{BanGPW16}. Thus, the corresponding value function is given by
 \small
 \begin{eqnarray}
   V(N-2,x) = xG_{N-2}(\bm{u}_{N-2}^{\ast}, \theta_{N-2}^{\ast}),   \nonumber
 \end{eqnarray}
 \normalsize
 where
 \small
 \begin{eqnarray}
              G_{N-2}(\bm{u}_{N-2}^{\ast}, \theta_{N-2}^{\ast})
   &  =  &  -  \theta_{N-2}^{\ast}\log\mathbb{E}\Big(\exp\big(-\bm{R}_{N-1}^\intercal\bm{u}_{N-2}^{\ast}\frac{1}{\theta_{N-2}^{\ast}}\big)\Big)  \nonumber \\
   &     &  +  e^{-\eta_{N-1}c_{N-1}\kappa_{N-1}}G_{N-1}(\bm{u}_{N-1}^{\ast}, \theta_{N-1}^{\ast})\times \nonumber \\
   &     &  \mathbb{E}\Big(\bm{R}_{N-1}^\intercal\bm{u}_{N-2}^{\ast}\Big) - \kappa_{N-2}S_{N-2}^{\ast} - \eta_{N-2}\theta_{N-2}^{\ast}.   \nonumber
 \end{eqnarray}
 \normalsize
 \textbf{Step 3}: Next, for $m = N-2, N-1, ..., 1,  0$, we use a backward  induction step. Assume that the claim holds for $m = n+1$. We need to show that it holds for  $m = n$.
 When $m = n$, the optimization problem becomes
 \small
 \begin{eqnarray}\label{eqn:step3op}
   &    \displaystyle\sup_{\bm{u}_{n} \in U}\inf_{\mathcal{E}_{n} \in \mathcal{B}_{\eta_{n}}}     &
     \Bigg(\mathbb{E}\big(\mathcal{E}_{n}\bm{R}_{n+1}^\intercal\bm{u}_{n}\big)
   + e^{-\eta_{n+1}c_{n+1}\kappa_{n+1}}G_{n+1}(\bm{u}_{n+1}^{\ast}, \theta_{n+1}^{\ast}) \mathbb{E}\big(\bm{R}_{n+1}^\intercal\bm{u}_{n}\big)    \nonumber \\
   &                                                                                              &
   -\kappa_{n}\sqrt{\bm{u}^\intercal_{n}\bm{\Sigma}_{n}\bm{u}_{n}}\Bigg). \ \ \ \ \ \
 \end{eqnarray}
 \normalsize
  The Lagrangian can be written as
 \small
 \begin{eqnarray}
             L_{n}(\mathcal{E}_{n},\theta_{n})
   &  =  &   \mathbb{E}\big(\mathcal{E}_{n}\bm{R}_{n+1}^\intercal\bm{u}_{n}\big) + e^{-\eta_{n+1}c_{n+1}\kappa_{n+1}}
             G_{n+1}(\bm{u}_{n+1}^{\ast}, \theta_{n+1}^{\ast})\times \nonumber \\
   &     &   \mathbb{E}\big(\bm{R}_{n+1}^\intercal\bm{u}_{n}\big) - \kappa_{n}\sqrt{\bm{u}^\intercal_{n}\bm{\Sigma}_{n}\bm{u}_{n}}
           + \theta_{n}\Big(\mathbb{E}\big(\mathcal{E}_{n}\log(\mathcal{E}_{n})\big) - \eta_{n}\Big).  \nonumber
 \end{eqnarray}
 \normalsize
 As in \textbf{Step 1} or \textbf{Step 2}, we obtain the optimal $\mathcal{E}_{n}^{\ast}$:
 \small
 \begin{eqnarray}\label{eqn:solKLop2}
   \mathcal{E}_{n}^{\ast}   &   =   &   \displaystyle\frac{\exp\big(-\bm{R}_{n+1}^\intercal\bm{u}_{n}\frac{1}{\theta_{n}}\big)}
                                         {\mathbb{E}\Big(\exp\big(-\bm{R}_{n+1}^\intercal\bm{u}_{n}\frac{1}{\theta_{n}}\big)\Big)} \ \ \ \mathbb{P}\textnormal{-a.s.},
 \end{eqnarray}
 \normalsize
 by solving
 \begin{eqnarray}
    \bm{u}^{T}_{n}\bm{R}_{n+1} + \theta_{n}\log(\mathcal{E}_{n}) + \theta_{n}    &   =   &   0, \ \ \  \mathbb{P}\textnormal{-a.s.}, \nonumber
 \end{eqnarray}
 together with the fact that the alternative distributions have a proper density, i.e.,
 \small
  \begin{eqnarray}
   \mathbb{E}(\mathcal{E}_{n})    &  =  &  1. \nonumber
 \end{eqnarray}
 \normalsize
 Moreover, the optimal $\mathcal{E}_{n}^{\ast}$ and its associated optimal $\theta_{n}$ satisfy the following system:
 \small
 \begin{eqnarray}
    \theta_{n}\Big(\mathbb{E}\big(\mathcal{E}_{n}\log(\mathcal{E}_{n})\big) - \eta_{n}\Big)   &   =   &  0, \nonumber \\
    \mathbb{E}\big(\mathcal{E}_{n}\log(\mathcal{E}_{n})\big)    &  \leq  &  \eta_{n}, \nonumber \\
                                                 \theta_{n}     &   >    &  0.    \nonumber
 \end{eqnarray}
 \normalsize
 We denote such solution as $\theta_{n}^{\ast}$. Now, with (\ref{eqn:solKLop2}), the optimization problem
 (\ref{eqn:step2op}) becomes
 \small
 \begin{eqnarray}\label{eqn:opuN-3}
    &    \displaystyle\sup_{\bm{u}_{n} \in U}   &
      \Bigg(-\theta_{n}^{\ast}\log\mathbb{E}\Big(\exp\big(-\bm{R}_{n+1}^\intercal\bm{u}_{n}\frac{1}{\theta_{n}^{\ast}}\big)\Big)
    + e^{-\eta_{n+1}c_{n+1}\kappa_{n+1}}\times           \nonumber \\
    &                                           &
G_{n+1}(\bm{u}_{n+1}^{\ast}, \theta_{n+1}^{\ast})\mathbb{E}\big(\bm{R}_{n+1}^\intercal\bm{u}_{n}\big)
    - \kappa_{n}\sqrt{\bm{u}^\intercal_{n}\bm{\Sigma}_{n}\bm{u}_{n}} - \eta_{n}\theta_{n}^{\ast}\Bigg).
 \end{eqnarray}
 \normalsize
 By Lemma B.1, we obtain the unique optimum of (\ref{eqn:opuN-3}) which satisfies the following system of nonlinear equations:
 \small
 \begin{eqnarray}
   \bm{u}^{\ast}_{n}    &   =   &    \frac{S_{n}^{\ast}}{\kappa_{n}}\Big(\bm{\Sigma}_{n}^{-1}\bm{X}^{\ast}_{n} - \frac{b_{n}^{\ast}\bm{\Sigma}_{n}^{-1}\bm{1}}{a_{n}}\Big)
                                   + \frac{\bm{\Sigma}_{n}^{-1}\bm{1}}{a_{n}}, \nonumber \\
         S^{\ast}_{n}   &   =   &    \sqrt{\frac{\frac{1}{a_{n}}}{1 - \frac{h_{n}}{\kappa_{n}^2} + \frac{(b_{n}^{\ast})^2}{\kappa_{n}^2 a_{n}}}}
                                   = \sqrt{\frac{\frac{1}{a_{n}}}{1 - \frac{1}{\kappa_{n}^2}g_{n}^{\ast}}}, \nonumber \\
      \bm{X}^{\ast}_{n} &   =   &    \frac{\mathbb{E}\Big(\exp(-\bm{R}^\intercal_{n+1}\bm{u}_{n}^{\ast}\frac{1}{\theta_{n}^{\ast}})\bm{R}_{n+1}\Big)}
                                     {\mathbb{E}\Big(\exp(-\bm{R}_{n+1}^\intercal\bm{u}_{n}^{\ast}\frac{1}{\theta_{n}^{\ast}})\Big)}      \nonumber  \\
                        &       & + e^{-\eta_{n+1}c_{n+1}\kappa_{n+1}}G_{n+1}(\bm{u}_{n+1}^{\ast}, \theta_{n+1}^{\ast})\mathbb{E}(\bm{R}_{n+1}). \nonumber
 \end{eqnarray}
 \normalsize
 The corresponding value function is given by
 \small
 \begin{eqnarray}
   V(n,x) = xG_{n}(\bm{u}_{n}^{\ast}, \theta_{n}^{\ast}),   \nonumber
 \end{eqnarray}
 \normalsize
 where
 \small
 \begin{eqnarray}
             G_{n}(\bm{u}_{n}^{\ast}, \theta_{n}^{\ast})
   &  =  & - \theta_{n}^{\ast}\log\mathbb{E}\Big(\exp\big(-\bm{R}_{n+1}^\intercal\bm{u}_{n}^{\ast}\frac{1}{\theta_{n}^{\ast}}\big)\Big)      \nonumber \\
   &     & + e^{-\eta_{n+1}c_{n+1}\kappa_{n+1}}G_{n+1}(\bm{u}_{n+1}^{\ast}, \theta_{n+1}^{\ast})\mathbb{E}\big(\bm{R}_{n+1}^\intercal\bm{u}_{n}^{\ast}\big)  \nonumber \\
   &     & - \kappa_{n}S_{n}^{\ast} - \eta_{n}\theta_{n}^{\ast}.   \nonumber
 \end{eqnarray}
 \normalsize
 This completes the proof.
\end{proof}

\section{Some Discussions of the Model.} \label{sec:ssproh}
In this section we discuss some modelling, theoretical and computational issues that may arise when we implement our approach. Also, we briefly discuss how to handle short
selling constraints and a generalization to the case of $\alpha$-divergence. \\

Firstly, it is worth noting that the uncertainty of the underlying distribution only enters into the expectation part, and the standard deviation part is added as a further
penalization. Mathematically, it is difficult to include the uncertainty in the standard deviation part as trying to do so leads to losing the time consistency property.
From a modelling and risk management perspective, since the error of estimation in the expectation part is far more serious than the standard deviation (see,  e.g.,
\cite{ChoZ93}), handling uncertainty in the expectation part is more important. \\

Secondly, as discussed in \cite{BanGPW16}, the strategy calculated in Theorem \ref{thm:opstrategysys} is optimal provided that the wealth stays positive. Of course, there
is no guarantee that this will always be the case. However, depending on risk tolerance, the investor may as well be happy to adopt such a strategy if the probability that the
wealth stays positive exceeds  a certain threshold. For more detailed discussions, we refer to \cite{BanGPW16}. To obtain such an optimal strategy, and to determine whether
the investor should  adopt such a   strategy, we modify \cite[algorithm 1]{BanGPW16}. This yields \textbf{Algorithm A}. It is worth noting that unlike algorithm 1, there
is no explicitly given lower bound on the risk aversion parameter $\kappa_{n}$. Instead, we constrain $\kappa_{n}$, so that the $\kappa_{n}$ chosen by the investor is a
valid risk aversion parameter in the sense that the system of nonlinear equations in Theorem \ref{thm:opstrategysys} is well defined.   \\

\begin{algorithm}[H]
  \caption{Multi-period MSD Robust Portfolio Selection}\label{algo:pss}
  \begin{center}
  \begin{algorithmic}[1]
   \scriptsize
    \State set abandon = false;
    \For{$n = N-1, ..., 0$}
        \State set $W_{n} = 1$ and select a $\kappa_{n} > 0$;
        \State solve  (\ref{eqn:opsolsys1}) and (\ref{eqn:opsolsys4}) simultaneously
               subject to $1-\frac{g_{n}^{\ast}}{\kappa_{n}^{2}} > 0$;
        \State calculate $p_{n}(\bm{u}_{n}) = \mathbb{P}\big(W_{n+1} > 0\big)$;
        \If{$p_{n}(\bm{u}_{n}) > 1 - \exp(-\kappa_{n})$}
          \State keep the strategy $\bm{u}_{n}$;
        \Else
          \State abandon = true;
        \EndIf
    \EndFor
    \If{abandon == false}
      \State take the investment;
    \Else
      \State abandon the investment;
    \EndIf
  \end{algorithmic}
  \end{center}
\end{algorithm}

In terms of computation, we see that to compute the robust strategy, we have to solve the system (\ref{eqn:opsolsys1}) - (\ref{eqn:opsolsys4}) simultaneously which requires
evaluation of the expectations in (\ref{eqn:opsolsys3}) - (\ref{eqn:opsolsys4}). It is almost impossible to evaluate such expectations directly in this system because of the
complicated interdependence of these equations. This difficulty can be resolved by applying a Monte Carlo type approach (see \cite[section 3]{GlaX14}). The idea is to
replace the theoretical expectations by sample means via simulations. In this way, we end up with a system of nonlinear equations which can then be solved numerically. \\

  It is worth noting that we have assumed that short selling is allowed. In portfolio selection, it is often required to impose short selling constraints. In our
case, this corresponds to replacing the set $U$ by
\small
\begin{eqnarray}
  U^{\textnormal{short}} = \Big\{\bm{u} \in \mathbb{R}^{d}: \bm{1}^\intercal\bm{u} = 1, \ \ell_{i} < \bm{u}^{i} < b_{i}, \ \textnormal{for} \ i = 1, ..., d, \ \textnormal{and}
                           \ \ell_{i}, b_{i} \in \mathbb{R}\Big\}. \nonumber
\end{eqnarray}
\normalsize
Indeed, by a straightforward Kuhn-Tucker argument (see, e.g., \cite{BoyV09}), it is easy to see that we can still find the optimal strategy without
losing the semi-analytical form under a short selling restriction. It is worth noting that if we change the strict inequality constraints in $U^{\textnormal{short}}$ to
inequality constraints, this adds a further difficulty and we may lose the semi-analytical form of the optimal strategy. \\

Also note the possibility of a risk free asset in the portfolio, which often occurs in practice and is of interest. Mathematically, it causes  difficulty in that the matrix $\Sigma_n$ becomes singular and non-invertible.  The work of \cite{LaMa12} deals with the presence of a risk free asset in  a non-robust single-period scenario. Specifically, their Corollary 1  points out that when a {\it{single}} constraint in the form $\bm{1}^\intercal\bm{u}=1$ is imposed and the risk aversion parameter $\kappa$ is large enough then only a trivial solution  exists.
The trivial solution implies that one should be fully invested in the risk free asset. Since in our case we are working under the condition of this single constraint, the inclusion of a risk free asset would not provide any further insights.\\

If further linear equality constraints  are imposed  then Theorem 1 in \cite{LaMa12} tells us that for a {\it{specific}} form of these  constraints, under {\it{specific}} assumptions on the distribution of returns {\it{and}} for a large enough $\kappa $ (that is, for a risk-averse investor), that a non-trivial solution can be obtained. This solution  corresponds to putting non-zero weights to both the risky and the risk free components. However, at this stage it is difficult to generalize Theorem 1 of \cite{LaMa12} for our multi-period robust portfolio selection  scenario and this generalization has been left as a future research agenda.
\\

Finally, it is well known that the KL divergence can be generalized to the so-called $\alpha$-divergence (see Appendix A). Our approach can be applied to such a case.
However,  several issues arise  when using $\alpha$-divergence. To have a properly defined worst case distribution,  the underlying random variable, i.e., the return of the assets,  must  be bounded (see \cite[proposition 2.3]{GlaX14}). Also, it is not clear to us whether the optimal strategy exists for all $\alpha$. To  properly handle  the uncertainty in the case where the underlying distribution is heavy tailed, another divergence measure may be needed. We are very interested in this case but it will
be dealt with in another paper.

\section{Numerical Examples.}\label{sec:numex}

In this section, we demonstrate the use of our model to select optimal strategy and to quantify model risk. Suppose our interest is to find the best allocation of a portfolio
of three stocks from the customer service industry--Navitas, Domino and Tabcorp--over an investment horizon of 5 days, i.e., $N = 5$. The historical daily prices
of these  stocks traded on the Australian Securities \mbox{Exchange}\footnote{Data obtained from Yahoo Finance https://au.finance.yahoo.com/} have been collected over the
period \mbox{1 Jan 2015 - 31 Dec 2015}. The corresponding daily returns  form a set of $261$  data points. In this section,
without loss of generality, a few assumptions will be made. The risk aversion parameter of the investor $\kappa_{n}$ is assumed to be constant and equals to three (i.e.,
$\kappa_{n} = 3$ for $n = 0 , 1, 2, 3, 4$). The initial wealth is assumed to be one dollar, i.e., $W_{0} = 1$. The random daily returns are assumed to be independent and
identically distributed over the investment horizon and have mean $\bm{\mu}$ and covariance matrix $\bm{\Sigma}$ under the nominal distribution.

\subsection{Comparison of Optimal Robust and Non-Robust Portfolio.}
Let us first consider a special case of uncertainty in distribution (i.e., the uncertainty in parameters) and compare the performance of the optimal robust and non-robust
strategies. The convenience of this simple scenario is the fact that the closed form formula for the KL divergence is sometimes available. For example, if the nominal
distribution is a $d$-dimensional multivariate normal distribution with mean $\bm{\mu}$ and covariance matrix $\bm{\Sigma}$, and the worst case distribution is a
$d$-dimensional multivariate normal distribution with mean $\bar{\bm{\mu}}$ and covariance matrix $\bar{\bm{\Sigma}}$, then the KL divergence can be calculated as (see, e.g.,
\cite[p. 296]{NieCD17}):
\small
\begin{eqnarray}\label{eqn:KLmulnormal}
 \mathcal{R}(\mathcal{E}) &  =  &   \frac{1}{2}\Bigg(trace(\bm{\Sigma}^{-1}\bar{\bm{\Sigma}}) + (\bm{\mu} - \bar{\bm{\mu}})^\intercal\bm{\Sigma}^{-1}(\bm{\mu} - \bar{\bm{\mu}})
                                  - d + \log\Big(\frac{|\bm{\Sigma}|}{|\bar{\bm{\Sigma}}|}\Big)\Bigg),
\end{eqnarray}
\normalsize
where $|\cdot|$ denotes the determinant of a matrix. For the purpose of illustration, we will consider the case where $\bar{\bm{\mu}} = \gamma\times\bm{\mu}$ for some $\gamma
\in \mathbb{R}$ and $\bar{\bm{\Sigma}} = \bm{\Sigma}$. Based on the collected data, we calculate the expected returns and the covariance matrix of the returns (under the nominal
distribution) as listed below:
\begin{eqnarray}
   &    & \scriptsize\bm{\mu}  = \left(\begin{array}{c}   0.0007  \\   0.0022  \\  0.0016  \end{array}\right), \ \
          \scriptsize\bm{\Sigma}  = \left(\begin{array}{ccc}         0.0003  &  0.0001  &  0.0001      \\
                                                                     0.0001  &  0.0004  &  0.0001      \\
                                                                     0.0001  &  0.0001  &  0.0003
                                  \end{array}\right). \nonumber
\end{eqnarray}
Formally, the $\bm{\mu}$ and $\bm{\Sigma}$ are estimates, however, for simplicity of notation, we do not use $\hat{\bm{\mu}}$ and $\hat{\bm{\Sigma}}$. \\

Since in the worst case scenario for the model disturbance, the alternative distribution is on the boundary of the KL divergence ball, then the divergence between the two models
is equal to $\eta_{n}$ (see Theorem \ref{thm:opstrategysys}). Without loss of generality, let us assume the radius of divergence ball $\eta_{n}$ to be constant over the
entire investment horizon. This allows us to find a time homogeneous $\gamma$. For simplicity, we drop the time dependence of $\eta_{n}$ and simply write $\eta$. In addition,
from now on, we will always choose $(c_{2}, c_{3}, c_{4}, c_{5})$ such that $(c_{2}\eta\kappa, c_{3}\eta\kappa, c_{4}\eta\kappa, c_{5}\eta \kappa) = (7.5,  8.0, 8.5,  9.0)$.
The choice of such numbers seems to be arbitrary  but one may consider these values as the  investor's risk tolerance for uncertainty of distribution (in contrast to $\kappa$
which is the risk aversion of the investor's preference for a fixed distribution). Thus, it solely depends on the investor's choice. As a consequence, the investor will have
 their own freedom to choose the amount of penalization (i.e. the effect of $(-c_{n}\eta\kappa)$) that they would like to take when selecting the portfolio. \\

\begin{table}[h]
\resizebox{1.0\textwidth}{!}{\begin{minipage}{\textwidth}
\tiny
\caption{Performance of robust and non-robust optimal solution: Comparison 1}\label{table:t1}
\centering
\begin{tabular}{cccc}
\hline
\multicolumn{1}{c}{$\gamma$}  & \multicolumn{1}{c}{$\eta$}  & \multirow{2}{*}{\shortstack[l]{number of times robust \\ outperforms non-robust}} & \multicolumn{1}{c}{\%}   \\
                              &                             &                                                                                  &                           \\
\hline
             0.2139           &            0.0050           &                                         244429                                   &          48.89\%\\

            -1.4859           &            0.0500           &                                         285828                                   &          57.17\%          \\

            -2.5156           &            0.1000           &                                         309814                                   &          61.96\%          \\

            -3.9718           &            0.2000           &                                         336583                                   &          67.32\%          \\

            -5.0892           &            0.3000           &                                         362909                                   &          72.58\%           \\

            -6.0312           &            0.4000           &                                         378952                                   &          75.79\%           \\

            -6.8611           &            0.5000           &                                         391459                                   &          78.29\%           \\
\hline
\end{tabular}
\end{minipage}}
\end{table}

\begin{table}[h]
\resizebox{1.0\textwidth}{!}{\begin{minipage}{\textwidth}
\tiny
\caption{Performance of robust and non-robust optimal solution: Comparison 2}\label{table:t2}
\centering
\begin{tabular}{cccccccc}
\hline
\multicolumn{1}{c}{$\gamma$}& \multicolumn{1}{c}{$\eta$}  &  \multicolumn{3}{c}{$\mathbb{E}(W_{N})$}  &  \multicolumn{3}{c}{$\frac{\mathbb{E}(W_{N}) - W_{0}}{\sqrt{Var(W_{N})}}$}  \\
\cline{3-5}
\cline{6-8}
                            &                             &    robust   &  non-robust  &  difference  &        robust       &       non-robust     &      difference         \\
\cline{1-8}
            0.2139          &            0.0050           &    1.0015   &    1.0016    &   -0.0001    &      \ 0.0515       &       \ 0.0542       &         -0.0027          \\

           -1.4859          &            0.0500           &    0.9903   &    0.9891    &  \ 0.0012    &       -0.3263       &        -0.3691       &        \ 0.0428         \\

           -2.5156          &            0.1000           &    0.9842   &    0.9816    &  \ 0.0026    &       -0.5267       &        -0.6277       &        \ 0.1010         \\

           -3.9718          &            0.2000           &    0.9761   &    0.9711    &  \ 0.0050    &       -0.7831       &        -0.9960       &        \ 0.2129        \\

           -5.0892          &            0.3000           &0.9715   &    0.9631    &  \ 0.0084    &       -0.9073       &        -1.2808       &        \ 0.3735          \\

           -6.0312          &            0.4000           &    0.9678   &    0.9564    &  \ 0.0114    &       -0.9965       &        -1.5224       &        \ 0.5259          \\

           -6.8611          &            0.5000           &    0.9650   &    0.9505    &  \ 0.0145    &       -1.0526       &        -1.7363       &         \ 0.6837          \\
\hline
\end{tabular}
\end{minipage}}
\end{table}

\begin{figure}[h]
  \centering
  \includegraphics[width=.5\textwidth]{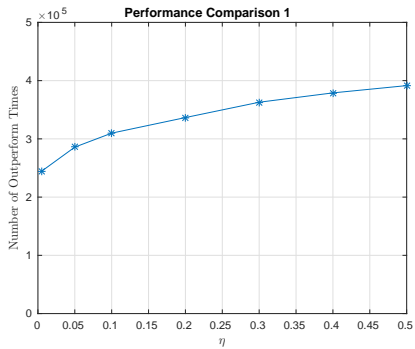}
  \caption{the number of times robust outperforms non-robust}
  \label{figure:f1}
\end{figure}

Now, suppose that we calculate under the worst case distribution. By generating data from this distribution, we compare the performance under the optimal robust and non-robust
strategies for different values of $\eta$. The optimal robust strategies are calculated by using $500,000$ Monte Carlo simulations. Then, we simulate $500,000$ daily return paths
(over $5$ days), and calculate the number of times, as well as the corresponding percentage, when the simulated terminal wealth under the robust case out-performs the non-robust
case (see \mbox{Table \ref{table:t1}}). \mbox{Figure \ref{figure:f1}} shows how the out-performance varies for different values of $\eta$. Other comparison metrics that we have
also calculated include the expected terminal wealth under both the robust and the non-robust case \footnote{the optimal non-robust strategy can be calculated by following
\cite{BanGPW16}.}, and the ratio of the difference between the expected terminal wealth and the initial wealth to the standard deviation of the terminal wealth (see
\mbox{Table \ref{table:t2}}). These are plotted in \mbox{Figure \ref{figure:f2}} and \mbox{Figure \ref{figure:f3}}, respectively. \\

\begin{figure}[h]
  \centering
  \includegraphics[width=.6\textwidth]{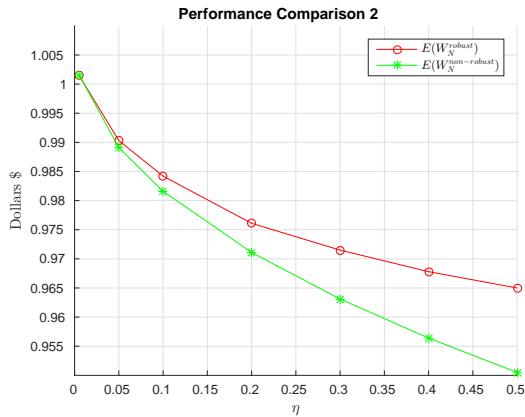}
  \caption{robust vs non-robust: expected terminal wealth}
  \label{figure:f2}
\end{figure}

\begin{figure}[h]
  \centering
  \includegraphics[width=.6\textwidth]{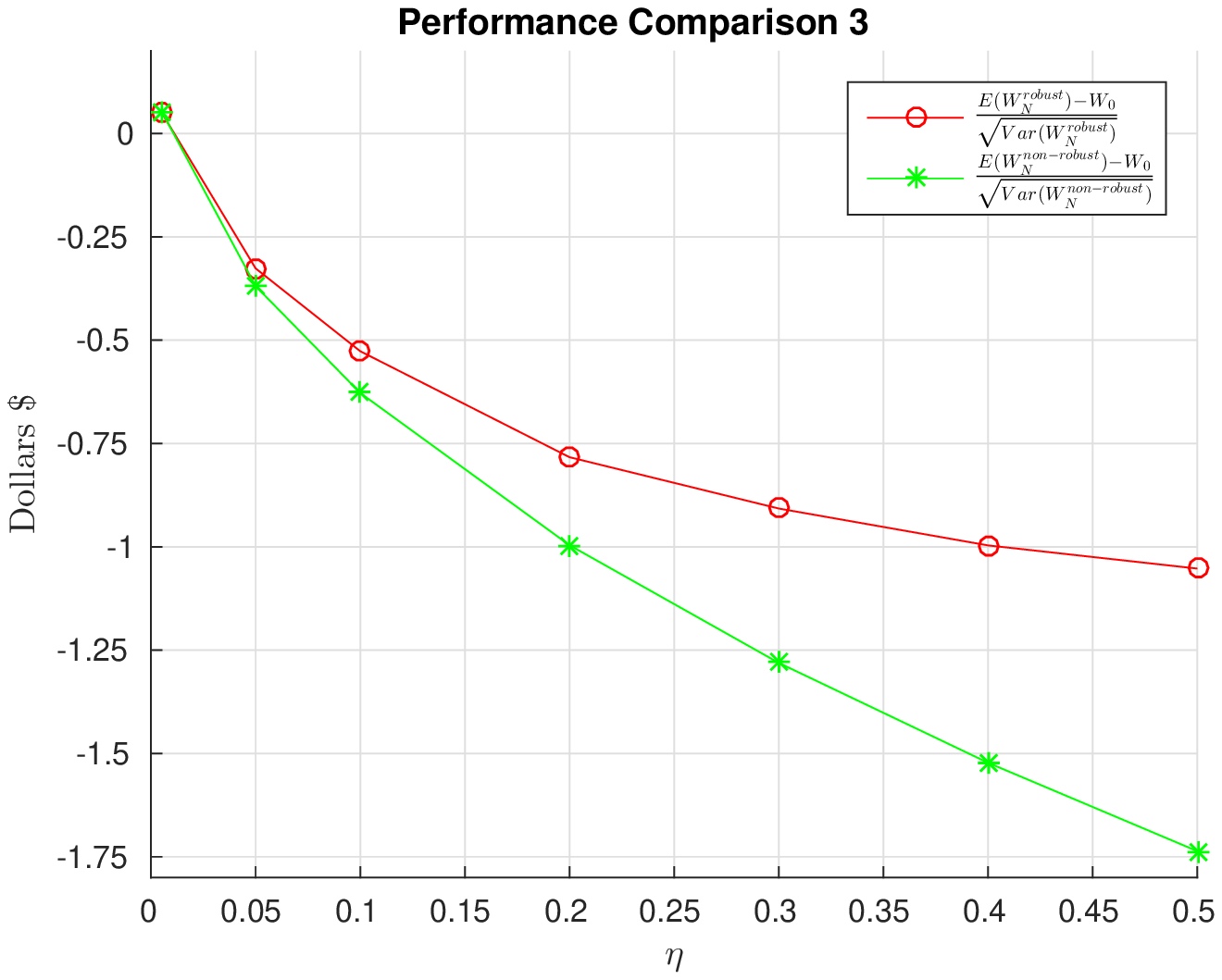}
  \caption{robust vs non-robust: ratio of the difference between the expected terminal wealth and the initial wealth to the standard deviation of the terminal wealth}
  \label{figure:f3}
\end{figure}

From Table \ref{table:t1}, we see that when $\eta$ is small, there is more than 50\% chance for the non-robust strategy to outperform the robust one. The corresponding expected
terminal wealth under the robust case is also higher than under the non-robust case. This suggests that if the worst case distribution is close to the nominal one (in the sense
of a small enough KL divergence), it may be hard to separate the two distributions and  thus it may be appropriate to use the nominal distribution. However, when $\eta$ is large,
it is clear that the robust strategy starts outperforming the non-robust one (with respect to each of the criteria that make sense in our discussion). This suggests that when the
worst case distribution is far from the nominal one, it is worth switching to the robust strategy. Furthermore, we notice that a profit is made when the radius is small and a loss is made when the radius is large. Thus, another suggestion could be that the optimal robust strategy is protecting against a loss of a portfolio due to the distribution
uncertainty and its impact is more apparent if a portfolio made a loss. \\

Another case where we have a closed form formula for the KL divergence is when both the nominal and alternative distributions are multivariate skew-normal. We will see in what
follows that in the extreme case where the nominal distribution degenerates to normal, we obtain similar comparison results as in the first example. \\

Let a nominal distribution be a multivariate skew-normal with location parameter $\bm{\mu}$, scale parameter $\bm{\Sigma}$ and skewness parameter $\bm{\xi}$. Next, we take the
worst case distribution to be a multivariate skew-normal with location parameter $\bar{\bm{\mu}}$, scale parameter $\bar{\bm{\Sigma}}$ and skewness parameter $\bar{\bm{\xi}}$.
The $d$-dimensional versions of these models are denoted by $\bm{Y} \sim SN_{d}(\bm{\mu},\bm{\Sigma},\bm{\xi})$ and $\bar{\bm{Y}} \sim SN_{d}(\bar{\bm{\mu}},\bar{\bm{\Sigma}},
\bar{\bm{\xi}})$ respectively. The closed form KL divergence is summarized in the following result,  for which the proof is in  \ref{appC}.
\begin{props}\label{props:skewnormalkl}
 Given a nominal distribution $\bm{Y} \sim SN_{d}(\bm{\mu},\bm{\Sigma},\bm{\xi})$ and an alternative distribution $\bar{\bm{Y}} \sim SN_{d}(\bar{\bm{\mu}},\bar{\bm{\Sigma}},
 \bar{\bm{\xi}})$, then the KL divergence between the nominal and the alternative distributions is given by:
 \small
 \begin{eqnarray}\label{eqn:KLskewmulnormal}
              \mathcal{R}_{skew}(\mathcal{E})
   &  =  &    \mathcal{R}(\mathcal{E}) + 2\sqrt{\frac{2}{\pi}}(\bar{\bm{\mu}}-\bm{\mu})^{\intercal}\bm{\Sigma}^{-1}\bar{\bm{\Sigma}}^{\frac{1}{2}}\bar{\bm{\xi}}
            - \mathbb{E}\Bigg(\log\Big(2\Phi\big(\Xi_{2}| 1 - \bm{\xi}^{T}\bm{\xi}\big)\Big)\Bigg)  \nonumber \\
   &     &  + \mathbb{E}\Bigg(\log\Big(2\Phi\big(\Xi_{1} | 1 - \bar{\bm{\xi}}^{T}\bar{\bm{\xi}}\big)\Big)\Bigg),
 \end{eqnarray}
 \normalsize
 where $\mathcal{R}(\mathcal{E})$ is given in (\ref{eqn:KLmulnormal}), $\Phi(\cdot|\sigma^{2})$ is the cumulative distribution function of a normal random variable with mean
 $0$ and variance $\sigma^{2}$, and
 \small
 \begin{eqnarray}
   \Xi_{1}  & \sim &  SN_{1}\Big(0,\bar{\bm{\xi}}^{\intercal}\bar{\bm{\xi}},\sqrt{\bar{\bm{\xi}}^{\intercal}\bar{\bm{\xi}}}\Big),    \nonumber \\
   \Xi_{2}  & \sim &  SN_{1}\Bigg(\bm{\xi}^{T}\bm{\Sigma}^{-\frac{1}{2}}(\bm{\mu} - \bar{\bm{\mu}}), \bm{\xi}^{T}
                     \bm{\Sigma}^{-\frac{1}{2}}\bar{\bm{\Sigma}}\bm{\Sigma}^{-\frac{1}{2}}\bm{\xi}, \frac{\bm{\xi}^{T}
                     \bm{\Sigma}^{-\frac{1}{2}}\bar{\bm{\Sigma}}^{\frac{1}{2}}\bar{\bm{\xi}}}{\sqrt{\bm{\xi}^{T}
                     \bm{\Sigma}^{-\frac{1}{2}}\bar{\bm{\Sigma}}\bm{\Sigma}^{-\frac{1}{2}}\bm{\xi}}}\Bigg).  \nonumber
 \end{eqnarray}
 \normalsize
\end{props}
It is worth noting that if both skewness parameters are equal to zero, we retain (\ref{eqn:KLmulnormal}). For more detailed discussions of (\ref{eqn:KLskewmulnormal}) and of
the multivariate skew-normal distribution, we refer to \cite{ConA12,AreG05}. \\

\begin{table}[h]
\resizebox{1.0\textwidth}{!}{\begin{minipage}{\textwidth}
\tiny
\caption{Performance of robust and non-robust optimal solution (skew-normal): Comparison 1}\label{table:t3}
\centering
\begin{tabular}{ccccc}
\hline
 \multicolumn{1}{c}{$\beta\%$}      &                                        $\xi$                                          & \multicolumn{1}{c}{$\eta$}  & \multirow{2}{*}{\shortstack[l]{number of times robust \\ outperforms non-robust}}&   \multicolumn{1}{c}{\%}   \\
                                    &                                                                                       &                             &                                                                                  &                            \\
\hline

              -78.30\%              &      $\left(\begin{array}{c} -0.0135 \\ -0.0982 \\  -0.0759 \end{array}\right)$       &            0.0050           &                                         244577                                   &         48.92\%            \\

             -242.54\%              &      $\left(\begin{array}{c} -0.0417 \\ -0.3042 \\  -0.2351  \end{array}\right)$      &            0.0500           &                                         284998                                   &         57.00\%             \\

             -334.61\%              &      $\left(\begin{array}{c} -0.0575 \\ -0.4196 \\  -0.3244  \end{array}\right)$      &            0.1000           &                                         307050                                   &         61.41\%             \\

             -449.99\%              &      $\left(\begin{array}{c} -0.0773  \\ -0.5643  \\ -0.4362 \end{array}\right)$      &            0.2000           &                                         329841                                   &         65.97\%\\

             -523.30\%              &      $\left(\begin{array}{c} -0.0899  \\  -0.6563 \\ -0.5073 \end{array}\right)$      &            0.3000           &                                         351005                                   &         70.20\%\\

             -572.42\%              &      $\left(\begin{array}{c} -0.0983 \\ -0.7178  \\ -0.5549 \end{array}\right)$       &            0.4000           &                                         361064                                   &         72.21\%\\

             -604.20\%              &      $\left(\begin{array}{c} -0.1038 \\ -0.7576 \\  -0.5857  \end{array}\right)$      &            0.5000           &                                         366798                                   &         73.36\%             \\
\hline
\end{tabular}
\end{minipage}}
\end{table}

For illustration, we take a $d$-dimensional multivariate normal distribution with mean $\bm{\mu}$ and covariance matrix $\bm{\Sigma}$ as the nominal distribution. The worst
case distribution is assumed to be a $d$-dimensional multivariate skew-normal distribution with a location parameter $\bar{\bm{\mu}}$, a scale parameter $\bar{\bm{\Sigma}}$
and a skewness parameter $\bar{\bm{\xi}}$, such that $\bar{\bm{\mu}}  =  \bm{\mu}$, and $\bar{\bm{\Sigma}}  =  \bm{\Sigma}$. We will choose $\bar{\bm{\xi}}$ such that the
mean of the worst case distribution is changed to $\beta\%$ of the mean of the nominal distribution. We note that the location parameter $\bm{\mu}$, and the scale parameter
$\bm{\Sigma}$ are \textit{not} the mean and the covariance matrix of the multivariate skew-normal. The reason that we choose this scenario is to show that the robust case is
indeed a strategy to safeguard against losses due to  a shift of the mean. This can be easily seen from the optimization procedure. Since we assume that $\eta$ is
constant over time,  we also have the same worst case distributions across time. \\

\begin{table}[h]
\resizebox{1.0\textwidth}{!}{\begin{minipage}{\textwidth}
\tiny
\caption{Performance of robust and non-robust optimal solution (skew-normal): Comparison 2}\label{table:t4}
\centering
\begin{tabular}{cccccccccc}
\hline
  \multicolumn{1}{c}{$\beta\%$}     &                                        $\xi$                                          & \multicolumn{1}{c}{$\eta$} & \multicolumn{3}{c}{$\mathbb{E}(W_{N})$} & \multicolumn{3}{c}{$\frac{\mathbb{E}(W_{N})- W_{0}}{\sqrt{Var(W_{N})}}$}  \\
\cline{4-9}
                                    &                                                                                       &                             &    robust   &   non-robust  &  difference  &      robust     &    non-robust    &     difference                 \\
\hline

              -78.30\%              &      $\left(\begin{array}{c} -0.0135 \\ -0.0982 \\  -0.0759 \end{array}\right)$       &            0.0050           &    1.0016   &     1.0016    &    0.0000    &     \ 0.0528    &    \ 0.0557      &       0.0029                   \\
\hline
             -242.54\%              &      $\left(\begin{array}{c} -0.0417 \\ -0.3042 \\  -0.2351  \end{array}\right)$      &            0.0500           &    0.9907   &     0.9896    &    0.0011    &      -0.3232    &     -0.3682      &       0.0450                    \\
\hline
             -334.61\%              &      $\left(\begin{array}{c} -0.0575 \\ -0.4196 \\  -0.3244  \end{array}\right)$      &            0.1000           &0.9853   &     0.9829    &    0.0024    &      -0.5187    &     -0.6316      &       0.1129                    \\
\hline
             -449.99\%              &      $\left(\begin{array}{c} -0.0773  \\ -0.5643  \\ -0.4362 \end{array}\right)$      &            0.2000           &0.9789   &     0.9745    &    0.0044    &      -0.7541    &     -1.0143      &       0.2602                    \\
\hline
             -523.30\%              &      $\left(\begin{array}{c} -0.0899  \\  -0.6563 \\ -0.5073 \end{array}\right)$      &            0.3000           &     0.9762   &     0.9692    &    0.0070    &      -0.8334    &     -1.3091      &       0.4757                    \\
\hline
             -572.42\%              &      $\left(\begin{array}{c} -0.0983 \\ -0.7178  \\ -0.5549 \end{array}\right)$       &            0.4000           &     0.9747   &     0.9657    &    0.0090    &      -0.8661    &     -1.5428      &       0.6767                    \\
\hline
             -604.20\%              &      $\left(\begin{array}{c} -0.1038 \\ -0.7576 \\  -0.5857  \end{array}\right)$      &            0.5000           &     0.9742   &     0.9634    &    0.0108    &      -0.8565    &     -1.7159      &       0.8594                    \\
\hline
\end{tabular}
\end{minipage}}
\end{table}

\begin{figure}[h]
  \centering
  \includegraphics[width=.5\textwidth]{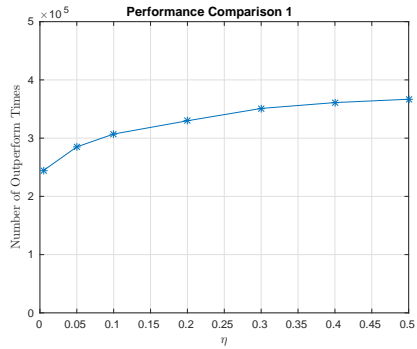}
  \caption{the number of times robust outperforms non-robust}
  \label{figure:f4}
\end{figure}

We again run $500,000$ simulations and calculate the number of times, as well as the corresponding percentage when the simulated terminal wealth under the robust case
outperforms the non-robust case for a range of divergences (see \mbox{Table \ref{table:t3}}). The divergences are calculated by (\ref{eqn:KLskewmulnormal}) using Monte
Carlo simulation with 500,000 simulations. \mbox{Figure \ref{figure:f4}} shows how the out-performance varies for different divergences. We also calculate the expected
terminal wealth under both the robust and the non-robust case, and the ratio of the difference between the expected terminal wealth and the initial wealth to the standard
deviation of the terminal wealth (see \mbox{Table \ref{table:t4}}). These have been plotted in \mbox{Figure \ref{figure:f5}} and \mbox{Figure \ref{figure:f6}},
respectively. \\

\begin{figure}[h]
  \centering
  \includegraphics[width=.6\textwidth]{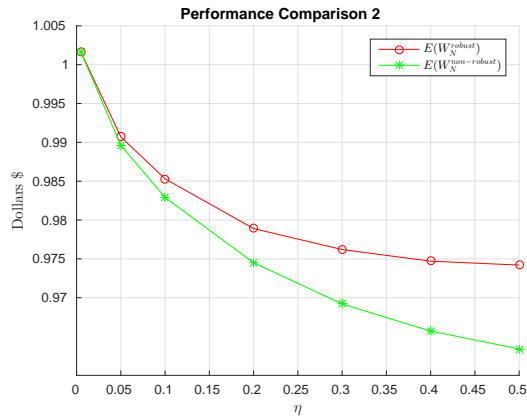}
  \caption{robust vs non-robust: expected terminal wealth}
  \label{figure:f5}
\end{figure}

\begin{figure}[h]
  \centering
  \includegraphics[width=.6\textwidth]{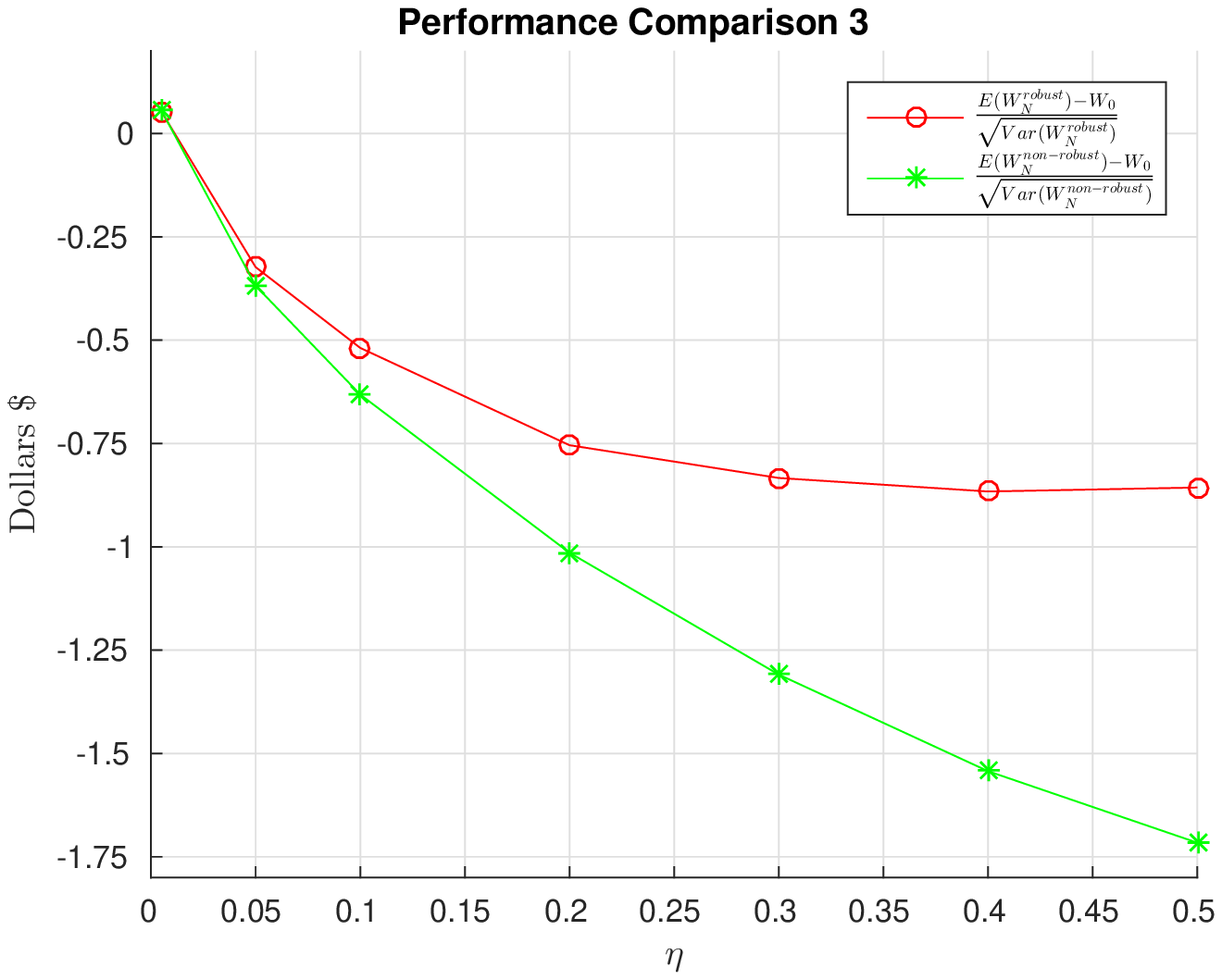}
  \caption{robust vs non-robust: ratio of the difference between the expected terminal wealth and the initial wealth to the standard deviation of the terminal wealth}
  \label{figure:f6}
\end{figure}


We notice that the performance comparison in \mbox{Figure \ref{figure:f5}} and \mbox{Figure \ref{figure:f6}} exhibits a similar pattern like in the previous example.
The number of times that the optimal robust strategy outperforms the non-robust strategy increases as the radius of the divergence increases. The values in the difference
columns in \mbox{Table \ref{table:t4}} also tend to increase as the radius of the KL ball increases. Hence, a larger uncertainty (i.e., a larger radius) makes the advantage
of the optimal robust strategy more apparent. As in the first example, we will choose the non-robust optimal strategy when the divergence is ``small" and choose the robust
optimal strategy when the divergence is ``large". \\

Up to this point, we only discussed when to choose an optimal non-robust strategy and when  to adopt  the robust one. It is worth noting that, although in some
cases it may be worth choosing the non-robust optimal strategy, we still need to quantify the amount of model risk involved in this action. In the next section, we define
model risk through the standard definition of risk in risk management, that is, as the quantile of a `profit-loss distribution' (see, e.g., \cite[p. 12]{ConGK10}). We also
provide a procedure to estimate the model risk by using empirical data. \\

\subsection{Quantification of Model Risk with Empirical Data.}

In this section, we discuss how to quantify model risk regarding  our optimal portfolio when only the empirical data is available and the true distribution is not
known. \\

Before we go into details, let us define model risk in terms of the quantile of a `profit-loss distribution'. Let $\mathbb{Q}$ denote the probability measure of a worst
case distribution, i.e., the empirical measure. The optimal portfolio is said to have a model risk of $\theta$ with a confidence level $q$ if
\small
\begin{eqnarray}
  \mathbb{Q}\Big(W^{non-robust}_{N} - W^{robust}_{N} \leq -\theta\Big)  &  =  &  1 - q.    \nonumber
\end{eqnarray}
\normalsize
In other words, we define model risk as the $(1 - q)$th-quantile of the distribution of the difference between the terminal wealth under the non-robust strategy and the robust
strategy. \\

Now, to quantify the model risk of our optimal portfolio, we divide the dataset into two subsets. The first subset, which we call dataset 1 (and it contains $201$ data points), is used
to estimate the expected value and the covariance matrix of the nominal distribution, which yields:

\begin{eqnarray}\label{eqn:empmeanandcov}
   &    & \scriptsize\tilde{\bm{\mu}} = \left(\begin{array}{c}         0.0009   \\  0.0019  \\  0.0014 \end{array}\right), \ \
          \scriptsize\tilde{\bm{\Sigma}} = \left(\begin{array}{ccc}    0.0003   &   0.0001   &   0.0001   \\
                                                                 0.0001   &   0.0003   &   0.0001   \\
                                                                 0.0001   &   0.0001   &   0.0002
                                  \end{array}\right).
\end{eqnarray}

The return distribution under the nominal distribution is again assumed to be a $d$-dimensional multivariate normal with mean and covariance matrix as given in
(\ref{eqn:empmeanandcov}). The second subset, which contains $60$ data points, is labelled as dataset 2. The distribution formed by taking equal probability for each
data point in dataset 2 is assumed to be a forecasted distribution of the incoming daily returns. We take this as the alternative distribution, and in the
following five days, assume that it is time-homogeneous, i.e., it is the same for each time period. \\

\begin{figure}[h]
   \center{}
   \includegraphics[width=0.65\textwidth]{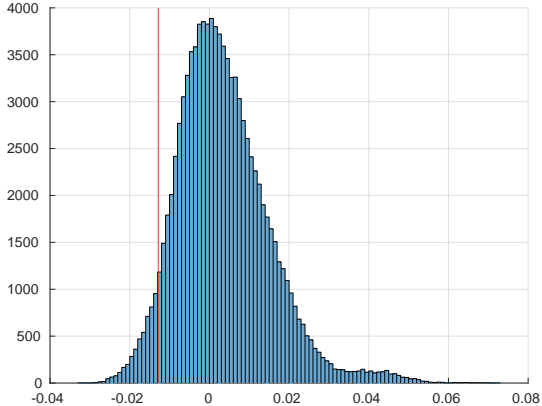}
   \caption{The distribution of $(W^{non-robust}_{N} - W^{robust}_{N})$} \label{figure:f7}
\end{figure}

\begin{algorithm}[h]
  \caption{Quantification of model risk by using dataset 1 and dataset 2}\label{algo:pss}
  \begin{center}
  \begin{algorithmic}[1]
    \scriptsize
    \State estimate $\tilde{\bm{\mu}}$ and $\tilde{\bm{\Sigma}}$ by using dataset 1;
    \For{$i = 1, 2, ..., 100,000$}
      \State generate $60$ independent sample from $N(\tilde{\bm{\mu}},\tilde{\bm{\Sigma}})$ (the normal sample);
      \State calculate ith $divergence$ by using (\ref{eqn:estkldiv}), the normal sample and dataset 2;
      \If{ith $divergence > 0$}
          \State keep the ith $divergence$;
        \Else
          \State abandon the ith $divergence$;
        \EndIf
    \EndFor
    \State estimate the divergence by taking the average of all the positive $divergence$ that have been calculated above;
    \State calculate the optimal robust strategy by using Theorem \ref{thm:opstrategysys} with $500,000$ simulations;
    \State calculate the optimal non-robust strategy;
    \For{$j = 1, ..., 100,000$}
        \State calculate $W^{non-robust}_{N} - W^{robust}_{N}$;
    \EndFor
    \State Find the $(1-q)$th-quantile of the distribution of $W^{non-robust}_{N} - W^{robust}_{N}$;
  \end{algorithmic}
  \end{center}
\end{algorithm}

Now, the next task is to estimate the divergence between the nominal distribution and the alternative distribution. We adapt an estimation procedure which is based on
the $k$th-nearest-neighbor approach (see, e.g., \cite{Per08, SchS15, WanKV09} for this method). Each time, a sample of $60$ data is generated from the nominal
distribution and the divergence between the nominal distribution and the alternative distribution is estimated by using this sample, the dataset 2, and Equation (\ref{eqn:estkldiv}).
To reduce variance, we repeat this procedure by taking the average over 100,000 repetitions (see, e.g., \cite{WanKV09}). Then, the estimated divergence is obtained as:
\small
\begin{eqnarray}
  \hat{\mathcal{R}}(\mathcal{E}) \approx  0.6455.  \nonumber
\end{eqnarray}
\normalsize
By knowing the KL divergence, we use a bootstrapping type of approach to sample $100,000$ data points from dataset 2. This allows us to construct the distribution of
$W^{non-robust}_{N} - W^{robust}_{N}$ (see \mbox{Figure \ref{figure:f7}}) from which the model risk can be estimated. The estimated model risk at $q = 95\%$ confidence
level is $0.0128$ (see the red vertical line in \mbox{Figure \ref{figure:f7}}). The interpretation is that if the optimal non-robust strategy is applied but the optimal
robust strategy turns out to be more appropriate, then $95\%$ of the time we would lose no more than $1.28$ cents for every one dollar. The complete procedure of quantifying
the model risk is summarized in \textbf{Algorithm B}.

\section{Conclusion.}
In this work, we have derived a semi-analytical form of an optimal robust strategy for an investment portfolio in which an uncertainty of distribution of return is involved.
The uncertainty is measured by the Kullback-Leibler divergence. We have applied our approach to several numerical examples and have suggested whether  to adopt the optimal
robust or non-robust strategy. In addition, we define model risk from the standard risk management perspective and present an algorithm for quantifying the model risk by using
empirical data. This delivers a convenient way of quantifying model risk in practice. \\

There are some possible variations and extensions of our work  that deserve further investigation.  One research question is about designing a  fair  way to perform an out-of-sample comparison of our method with the non-robust method. The purpose would be to compare  the
performance of the robust and the non-robust strategy directly rather than via assessing  the risk of applying  the non-robust strategy that  we have presented here. Another research question  is
to scrutinize  the cases where the size of the portfolio is very large (for example, $d > 100$). We have performed some initial testing of the performance of the solution
for portfolios  up to size 50  but do not  currently  have theoretical criteria to  guarantee the existence of a solution of the system in Theorem \ref{thm:opstrategysys}.
In addition, it is  worth investigating  the computational cost  as the problem scales to very large portfolio sizes. \\

\appendix

\section{Kullback-Leibler Divergence: Concept, Extension, and Estimation.}
The Kullback-Leibler (KL) divergence is a well known deviation measure between distributions, and has been discussed in many papers. Here, we present a brief summary of its
concept, extension and estimation by mainly following \cite{BreC16,GlaX14,SchS15}. \\

Fix a probability space $(\Omega, \mathcal{F}, \mathbb{P})$, and let us assume that the nominal distribution is described by the probability measure $\mathbb{P}$, and denote
its density by $f$. An alternative distribution described by a probability measure $\mathbb{Q}$ is assumed to be absolutely continuous with respect to $\mathbb{P}$, i.e., for
$A \in \mathcal{F}$ and $\mathbb{P}(A) = 0$, we have $\mathbb{Q}(A) = 0$. The density of an alternative distribution is denoted by $g$. Let \small$\displaystyle\mathcal{E} =
\frac{g}{f}$. \normalsize The KL divergence between the nominal and the alternative distribution is defined as
\scriptsize
\begin{eqnarray*}
 \mathcal{R}(\mathcal{E}) &   =   &  \mathbb{E}\Big(\mathcal{E}\log\mathcal{E}\Big)
                             :=      \mathbb{E}\Big(\mathcal{E}(\bm{\zeta})\log\mathcal{E}(\bm{\zeta})\Big)
                              =      \int \Big(\frac{g(\bm{\zeta})}{f(\bm{\zeta})}\Big)\log\Big(\frac{g(\bm{\zeta})}{f(\bm{\zeta})}\Big)f(\bm{\zeta})d\bm{\zeta}.
\end{eqnarray*}
\normalsize
The KL divergence can be generalized to the so-called $\alpha$-divergence. For $\alpha > 1$, the $\alpha$-divergence is defined as
\scriptsize
\begin{eqnarray*}
 \mathcal{R}(\mathcal{E})     =      \frac{\mathbb{E}\big(\mathcal{E}^{\alpha} - \alpha(\mathcal{E} - 1) - 1\big)}{\alpha(\alpha - 1)}
                          &  :=   &    \frac{\mathbb{E}\big((\mathcal{E}(\bm{\zeta}))^{\alpha} - \alpha(\mathcal{E}(\bm{\zeta}) - 1) - 1\big)}{\alpha(\alpha - 1)}  \\
                          &   =   &   \frac{\displaystyle\int\Bigg(\Big(\frac{g(\bm{\zeta})}{f(\bm{\zeta})}\Big)^{\alpha}
                                     - \alpha\Big(\Big(\frac{g(\bm{\zeta})}{f(\bm{\zeta})}\Big) - 1\Big) -  1\Bigg)f(\bm{\zeta})d\bm{\zeta}}{\alpha(\alpha - 1)}.
\end{eqnarray*}
\normalsize
By applying the L'H\^{o}spital's rule, one can show that the $\alpha$-divergence converges to the KL divergence as $\alpha \rightarrow 1$. In other words, the KL divergence
is the limiting case of the $\alpha$-divergence. \\

When we know the nominal and the alternative distributions precisely, sometimes we may be able to calculate the KL divergence analytically. In practice, however, we often
only have samples of the distributions. This requires estimation of the KL divergence. One way to estimate the KL divergence is by using the $k$th-nearest-neighbor
estimation approach. Suppose we have an independent identically distributed (i.i.d) sample $(\bm{Y}_{i})$ from the nominal distribution and another i.i.d sample
$(\tilde{\bm{Y}}_{i})$ from the alternative distribution. The estimated KL divergence between the two models by using the $k$th-nearest-neighbor approach is given by
\scriptsize
\begin{eqnarray}\label{eqn:estkldiv}
  \hat{\mathcal{R}}(\mathcal{E}) = \frac{1}{K}\sum_{i = 1}^{K}\log\Big(\frac{K (y_{k}(i))^{d}}{(K-1)(\tilde{y}_{k}(i))^{d}}\Big).
\end{eqnarray}
\normalsize
where $K$ is the sample size, $\tilde{y}_{k}(i)$ is the Euclidean distance of the $k$th-nearest-neighbor of $\tilde{\bm{Y}}_{i}$ in $(\tilde{\bm{Y}}_{j})_{j \neq i}$,
$y_{q}(i)$ is the Euclidean distance of the $k$th-nearest-neighbor of $\tilde{\bm{Y}}_{i}$ in $(\bm{Y}_{i})$, and $d$ is the dimension of the sample. In our case $d = 3$,
i.e., the number of assets in the portfolio. A simple rule to select $k$ is to take a small odd number. In this paper, we take $k = 5$. For more detailed discussion of the
$k$th-nearest-neighbor approach, we refer to \cite{SchS15,WanKV09}.

\section{A Useful Lemma.}
\begin{thm}[Lemma B.1]\label{lemma:convavity}
  The function
  \small
  \begin{eqnarray}
    h: \bm{u} \in U \rightarrow \Big(-\theta\log\mathbb{E}\Big(\exp\big(-\bm{u}^\intercal\bm{R}\frac{1}{\theta}\big)\Big)
                 - \kappa\sqrt{\bm{u}^\intercal\bm{\Sigma}\bm{u}}\Big)   \nonumber
  \end{eqnarray}
  \normalsize is strictly concave, where $\bm{R} \in \mathbb{R}^{d}$ is a random vector, $\kappa, \theta > 0$, and $\bm{\Sigma}$ is positive definite.
\end{thm}

\begin{proof}
 The proof that $-\kappa\sqrt{\bm{u}^\intercal\bm{\Sigma}\bm{u}}$
 is strictly concave follows from \cite[p. 4430]{Owa12}. Thus, it is sufficient to prove that the other part is concave. The latter follows from H\"{o}lder's inequality. Indeed, for
 $\bm{u}, \bm{v} \in \mathbb{R}^{d}$, and $t \in (0,1)$, we see that
 \scriptsize
 \begin{eqnarray}
   &        &   \mathbb{E}\Big(\exp\big(-(t\bm{u} + (1-t)\bm{v})^\intercal\bm{R}\frac{1}{\theta}\big)\Big)      \nonumber   \\
   &   =    &   \mathbb{E}\Big(\exp\big(-t\bm{u}^\intercal\bm{R}\frac{1}{\theta}\big)\exp\big(-(1-t)\bm{v}^\intercal\bm{R}\frac{1}{\theta}\big)\Big)        \nonumber \\
   &  \leq  &   \Bigg(\mathbb{E}\Big(\exp\big(-\bm{u}^\intercal\bm{R}\frac{1}{\theta}\big)\Big)\Bigg)^{t}
                \Bigg(\mathbb{E}\Big(\exp\big(-\bm{v}^\intercal\bm{R}\frac{1}{\theta}\big)\Big)\Bigg)^{1-t}. \nonumber
 \end{eqnarray}
 \normalsize
 By taking logarithms of both sides and by multiplying by $-\theta ,$ we obtain the desired result.
\end{proof}

\section{Proof of Proposition 1.}\label{appC}
In order to prove Proposition \ref{props:skewnormalkl}, we start with some basics from \cite{AreG05,ConA12}.


Given $\bm{Y} \sim SN_{d}(\bm{\mu},\bm{\Sigma},\bm{\xi})$, the density of $\bm{Y}$ is given by
\begin{eqnarray}\label{eqn:skewnormalpdf}
  f_{\bm{Y}}(\bm{y})  &  =  &  2|\bm{\Sigma}^{-\frac{1}{2}}|\phi_{d}\Big(\bm{\Sigma}^{-\frac{1}{2}}(\bm{y}-\bm{\mu})\Big)\Phi\Big(\bm{\xi}^{\intercal}\bm{\Sigma}^{-\frac{1}{2}}(\bm{y}-\bm{\mu})\Big|1-\bm{\xi}^{\intercal}\bm{\xi}\Big),
\end{eqnarray}
where $\phi_{d}$ is the density of a $d$-dimensional standard normal, and $\Phi(\cdot|\sigma^{2})$ is the cumulative distribution function of a standard normal with mean $0$ and variance $\sigma^{2}$. \\

Moreover, we have
\begin{eqnarray}\label{eqn:skewnormalmstd}
  \mathbb{E}\big(\bm{Y}\big) = \bm{\mu} + \sqrt{\frac{2}{\pi}}\bm{\Sigma}^{\frac{1}{2}}\bm{\xi} \ \ \ \textnormal{and} \ \ \
  Var\big(\bm{Y}\big) = \bm{\Sigma} - \frac{2}{\pi}\bm{\Sigma}^{\frac{1}{2}}\bm{\xi}\bm{\xi}^{\intercal}\bm{\Sigma}^{\frac{1}{2}}.
\end{eqnarray}
and
\begin{eqnarray}
   \bm{Y}   &           =         &   \bm{\mu} +  \bm{\Sigma}^{\frac{1}{2}} \bm{Y}^{\ast}  \stackrel{d}{=}  \bm{\mu} +  \bm{\Sigma}^{\frac{1}{2}} \Big(\bm{\xi}|Z_{0}| + (I_{d} - \bm{\xi}\bm{\xi}^{\intercal})^{\frac{1}{2}}\bm{Z}\Big),    \nonumber
\end{eqnarray}
where $Z_{0} \sim N(0,1)$ and $\bm{Z} \sim N_{d}(\bm{0},\bm{I}_{d})$ are independent one-dimensional Normal and $d$-dimensional Multivariate Normal distributions. \\

If the nominal model is $\bm{Y} \sim SN_{d}(\bm{\mu},\bm{\Sigma},\bm{\xi})$ and the alternative model is $\bar{\bm{Y}} \sim SN_{d}(\bar{\bm{\mu}},\bar{\bm{\Sigma}},\bar{\bm{\xi}})$, the KL divergence between the two models is given by
\begin{eqnarray}
   \mathcal{R}_{skew}(\mathcal{E})   &  =  &    C(\bm{Y}, \bar{\bm{Y}}) - C(\bar{\bm{Y}}, \bar{\bm{Y}}), \nonumber
\end{eqnarray}
where
\begin{eqnarray}\label{eqn:crossentropy}
  C(\bm{Y}, \bar{\bm{Y}}) &  =  & - \mathbb{E}\Big(\log\big(f_{Y}(\bar{\bm{Y}})\big)\Big)
\end{eqnarray}
is the cross-entropy (see \cite[p. 14]{ConA12}). \\

Next, we proceed to the  proof of Proposition \ref{props:skewnormalkl}.

\begin{proof}
  Since $\bm{Y} \sim SN_{d}(\bm{\mu},\bm{\Sigma},\bm{\xi})$, by (\ref{eqn:skewnormalpdf}), it is easy to see that
  \begin{eqnarray}
    \log(f_{Y}(\bm{y}))  &  =  &    \log\Big(|\bm{\Sigma}|^{-\frac{1}{2}}\phi_{d}\big(\bm{\Sigma}^{-\frac{1}{2}}(\bm{y}-\bm{\mu})\big)\Big) \nonumber \\
                         &     &  + \log\Big(2\Phi\big(\bm{\xi}^{T}\bm{\Sigma}^{-\frac{1}{2}}(\bm{y}-\bm{\mu}) \big| 1 - \bm{\xi}^{T}\bm{\xi}\big)\Big) \nonumber \\
                         &  =  &  - \frac{1}{2}\Big(\log(|\bm{\Sigma}|^{-1}) + d\log(2\pi) +(\bm{y} - \bm{\mu})^{\intercal}\bm{\Sigma}^{-1}(\bm{y} - \bm{\mu})\Big)    \nonumber \\
                         &     &  + \log\Big(2\Phi\big(\bm{\xi}^{T}\bm{\Sigma}^{-\frac{1}{2}}(\bm{y}-\bm{\mu}) \big| 1 - \bm{\xi}^{T}\bm{\xi}\big)\Big) \nonumber
  \end{eqnarray}
  The cross-entropy between $\bm{Y} \sim SN_{d}(\bm{\mu},\bm{\Sigma},\bm{\xi})$ and $\bar{\bm{Y}} \sim SN_{d}(\bar{\bm{\mu}},\bar{\bm{\Sigma}},\bar{\bm{\xi}})$ is then given by
  \begin{eqnarray}
    C(\bm{Y}, \bar{\bm{Y}}) &  =  & - \mathbb{E}\Big(\log\big(f_{Y}(\bar{\bm{Y}})\big)\Big)             \nonumber \\
                            &  =  &   \frac{1}{2}\Bigg(\log(|\bm{\Sigma}|^{-1}) + d\log(2\pi) + \mathbb{E}\Big((\bar{\bm{Y}} - \bm{\mu})^{\intercal}\bm{\Sigma}^{-1}(\bar{\bm{Y}}- \bm{\mu})\Big)\Bigg)  \nonumber
  \end{eqnarray}
  \begin{eqnarray}
                            &     & - \mathbb{E}\Bigg(\log\Big(2\Phi\big(\bm{\xi}^{T}\bm{\Sigma}^{-\frac{1}{2}}(\bar{\bm{Y}}-\bm{\mu}) \big| 1 - \bm{\xi}^{T}\bm{\xi}\big)\Big)\Bigg). \nonumber
  \end{eqnarray}
  As a consequence of \cite[part (iii) of Lemma 1]{ConA12} and (\ref{eqn:skewnormalmstd}), we obtain
  \begin{eqnarray}
                \mathbb{E}\Big((\bar{\bm{Y}} - \bm{\mu})^{\intercal}\bm{\Sigma}^{-1}(\bar{\bm{Y}} - \bm{\mu})\Big)\Bigg)
    &  =  &   tr\Big(\bm{\Sigma}^{-1}\bar{\bm{\Sigma}}\Big) + (\bar{\bm{\mu}} - \bm{\mu})^{\intercal}\bm{\Sigma}^{-1}(\bar{\bm{\mu}}-\bm{\mu})       \nonumber  \\
    &     & + 2\sqrt{\frac{2}{\pi}}(\bar{\bm{\mu}}-\bm{\mu})^{\intercal}\bm{\Sigma}^{-1}\bar{\bm{\Sigma}}^{\frac{1}{2}}\bar{\bm{\xi}},               \nonumber
  \end{eqnarray}
  which implies
  \begin{eqnarray}
    &     &    C(\bm{Y}, \bar{\bm{Y}})                                                                                                                                                                                \nonumber \\
    &  =  &    \frac{1}{2}\Bigg(\log(|\bm{\Sigma}|^{-1}) + d\log(2\pi) +  tr\Big(\bm{\Sigma}^{-1}\bar{\bm{\Sigma}}^{-1}\Big) + (\bar{\bm{\mu}} - \bm{\mu})^{\intercal}\bm{\Sigma}^{-1}(\bar{\bm{\mu}}-\bm{\mu})\Bigg) \nonumber \\
    &     &  + 2\sqrt{\frac{2}{\pi}}(\bar{\bm{\mu}}-\bm{\mu})^{\intercal}\bm{\Sigma}^{-1}\bar{\bm{\Sigma}}^{\frac{1}{2}}\bar{\bm{\xi}}
             - \mathbb{E}\Bigg(\log\Big(2\Phi\big(\bm{\xi}^{T}\bm{\Sigma}^{-\frac{1}{2}}(\bar{\bm{Y}}-\bm{\mu}) \big| 1 - \bm{\xi}^{T}\bm{\xi}\big)\Big)\Bigg). \nonumber
  \end{eqnarray}

  Since $\bar{\bm{Y}} \sim SN(\bar{\bm{\mu}},\bar{\bm{\Sigma}},\bar{\bm{\xi}})$, this yields
  \begin{eqnarray}
    \Xi_{2}     &        =          &  \bm{\xi}^{T}\bm{\Sigma}^{-\frac{1}{2}}(\bar{\bm{Y}} - \bm{\mu})                                                                                                           \nonumber \\
                &        =          &  \bm{\xi}^{T}\bm{\Sigma}^{-\frac{1}{2}}(\bar{\bm{\mu}} - \bm{\mu}) + \bm{\xi}^{T}\bm{\Sigma}^{-\frac{1}{2}}\bar{\bm{\Sigma}}^{\frac{1}{2}}\bar{\bm{Y}}^{\ast}              \nonumber  \\
                &  \stackrel{d}{=}  &  \bm{\xi}^{T}\bm{\Sigma}^{-\frac{1}{2}}(\bar{\bm{\mu}} - \bm{\mu}) + \bm{\xi}^{T}\bm{\Sigma}^{-\frac{1}{2}}\bar{\bm{\Sigma}}^{\frac{1}{2}}\bar{\bm{\xi}}|Z_{0}|            \nonumber \\
                &                   & + \sqrt{\bm{\xi}^{\intercal}\bm{\Sigma}^{-\frac{1}{2}}\bar{\bm{\Sigma}}\bm{\Sigma}^{-\frac{1}{2}}\bm{\xi} - \big(\bm{\xi}^{\intercal}\bm{\Sigma}^{-\frac{1}{2}}\bar{\bm{\Sigma}}^{\frac{1}{2}}\bar{\bm{\xi}}\big)^{2}}Z_{1} \nonumber \\
                &        \sim       &  SN\Bigg(\bm{\xi}^{T}\bm{\Sigma}^{-\frac{1}{2}}(\bar{\bm{\mu}} - \bm{\mu}), \bm{\xi}^{T}\bm{\Sigma}^{-\frac{1}{2}}\bar{\bm{\Sigma}}\bm{\Sigma}^{-\frac{1}{2}}\bm{\xi}, \frac{\bm{\xi}^{T}\bm{\Sigma}^{-\frac{1}{2}}\bar{\bm{\Sigma}}^{\frac{1}{2}}\bar{\bm{\xi}}}{\sqrt{\bm{\xi}^{T}\bm{\Sigma}^{-\frac{1}{2}}\bar{\bm{\Sigma}}\bm{\Sigma}^{-\frac{1}{2}}\bm{\xi}}}\Bigg),  \nonumber \\
                &                   &  \textnormal{if} \ \bm{\xi} \ \textnormal{is not} \ \bm{0}, \nonumber
  \end{eqnarray}
 and $\Xi_{2} = 0$ otherwise, where $Z_{1} \sim N(0,1)$ is independent of $Z_{0}$. This then implies
 \begin{eqnarray}
    C(\bm{Y}, \bar{\bm{Y}}) &  =  &    C(\bm{Y}_{0}, \bar{\bm{Y}}_{0})
                                     + 2\sqrt{\frac{2}{\pi}}(\bar{\bm{\mu}}-\bm{\mu})^{\intercal}\bm{\Sigma}^{-1}\bar{\bm{\Sigma}}^{\frac{1}{2}}\bar{\bm{\xi}}    \nonumber  \\
                            &     &  - \mathbb{E}\Bigg(\log\Big(2\Phi\big(\Xi_{2} \big| 1 - \bm{\xi}^{T}\bm{\xi}\big)\Big)\Bigg).                                 \nonumber
 \end{eqnarray}
 where $\bm{Y}_{0} \sim SN_{d}(\bm{\mu},\bm{\Sigma},0)$ and $\bar{\bm{Y}}_{0} \sim SN_{d}(\bar{\bm{\mu}},\bar{\bm{\Sigma}},0)$. \\

 Since
 \begin{eqnarray}
   \mathcal{R}_{skew}(\mathcal{E})    &  =  & C(\bm{Y}, \bar{\bm{Y}}) - C(\bar{\bm{Y}}, \bar{\bm{Y}}), \nonumber
 \end{eqnarray}
 after some simple algebra we obtain the desired result.
\end{proof}

\section*{Acknowledgments.} This work was supported by the Australian Research Council's Discovery Project funding scheme (Project DP160103489). The authors are grateful to the Editor and to the referees for the constructive criticisms.

\bibliographystyle{model5-names}

\begin{thebibliography}{17}

\bibitem[{Arellano-Valle and Genton(2005)}]{AreG05}
Arellano-Valle, R. B. and Genton, M. G. (2005). On fundamental skew distributions.
Journal of Multivariate Analysis 96~(1), 93 -- 116.

\bibitem[{Bannister et~al.(2016) Bannister and Goldys and Penev and Wu}]{BanGPW16}
Bannister, H. \& Goldys, B. \& Penev, S.  \&  Wu, W. (2016) Multiperiod mean-standard-deviation time consistent portfolio selection.
Automatica 73, 15 -- 26

\bibitem[{Ben-Tal et~al.(1988) Ben-Tal and Teboulle and Charnes}]{BenTC88}
Ben-Tal, A. \& Teboulle, M. \& Charnes A. (1988). The role of duality in optimization problems involving entropy functionals with applications to information theory.
Journal of Optimization Theory and Applications 58~(2), 209 -- 223.

\bibitem[{Boyd and Vandenberghe(2004)}]{BoyV09}
Boyd, S. \& Vandenberghe, L. (2004). Convex Optimization. (Seventh printing with corrections 2009)
Cambrisge University Press

\bibitem[{Breuer and Csisz\'{a}r(2016)}]{BreC16}
Breuer, T. \& Csisz\'{a}r, I. (2016). Measuring distribution model risk.
Mathematical Finance 26~(2), 395 -- 411.

\bibitem[{Calafiore(2007)}]{Cal07}
Calafiore, G. C. (2007). Ambiguous risk measures and optimal robust portfolio.
SIAM Journal of Optimization 18~(3), 853 -- 877.

\bibitem[{Chen et~al.(2013) Chen and Li and Guo(2013)}]{CheLG13}
Chen, Z. P. \& Li, G. \& Guo, J. E. (2013). Optimal investment policy in the time consistent mean-variance formulation.
Insurance: Mathematics and Economics 52~(2), 145 -- 156.

\bibitem[{Chopra  and Ziemba(1993)}]{ChoZ93}
Chopra, V. K. \& Ziemba, W. T. (1993). The effect of errors in means, variances, and covariances on optimal portfolio choice.
Journal of Portfolio Management, 19,  6 -- 12.

\bibitem[{Connor et~al.(2010) Conner and Goldberg and Korajczyk}]{ConGK10}
Connor, G. \& Goldberg, L. R. \& Korajczyk, R. A. (2010). Portfolio Risk Analysis.
Princeton University Press

\bibitem[{Contreras-Reyes and Arellano-Valle(2012)}]{ConA12}
Contreras-Reyes, J. E. and Arellano-Valle, R. B. (2012). Kullback-{L}eibler divergence measure for multivariate skew-normal distributions.
Entropy 14~(9), 1606 -- 1626.


\bibitem[{Glasserman and Xu(2013)}]{GlaX13}
Glasserman, P. \& Xu, X. B. (2013). Robust portfolio control with stochastic factor dynamics.
Operations Research 61~(4), 874 -- 893.

\bibitem[{Glasserman and Xu(2014)}]{GlaX14}
Glasserman, P. \& Xu, X. B. (2014). Robust Risk Measurement and Model Risk.
Quantitative Finance 14~(1), 29 -- 58.

\bibitem[{Kang and Filar(2006)}]{KanF06}
Kang, B. D. \& Filar, J. (2006). Time consistent dynamic risk measures.
Mathematical Methods in Operations Research 63~(1), 169 -- 186.

\bibitem[{Kapsos et~al.(2014) Kapsos and Christofides and Rustem}]{KapCR14}
Kapsos, M. \& Christofides, N. \& Rustem, B. (2014). Worst-case robust Omega ratio.
European Journal of Operational Research 234~(2), 499 -- 507.

\bibitem[{Kim et~al.(2014) Kim and Kim and Kim and Fabozzi}]{KimKKF14}
Kim, W. C. \& Kim, M. J. \& Kim, J. H. \&  Fabozzi, F. J. (2014). Robust portfolios that do not tilt factor exposure.
European Journal of Operational Research 234~(2), 411 -- 421.

\bibitem[{Lam(2016)}]{Lam16}
Lam, H. (2016). Robust sensitivity analysis for stochastic Systems.
Mathematics of Operations Research 41~(4), 1248 -- 1275.

\bibitem[{Landsman and Makov(2012)}]{LaMa12}  Landsman, Z. \& Makov, U. (2012). Translation-invariant and positive-homogeneous risk measures and optimal portfolio management
in the presence of a riskless component. {\it Insurance: Mathematics and Economics}, 50(1), 94-98.

\bibitem[{Li and Ng(2000)}]{LiN00}
Li, D. and Ng, W. L. (2000). Optimal Dynamic Portfolio Selection: Multiperiod Mean-Variance Formulation.
Mathematical Finance 10~(3),  387 -- 406.

\bibitem[{Markowitz(1952)}]{Mar52}
Markowitz, H. (1952). Portfolio Selection.
The Journal of Finance 7~(1), 77 -- 91.

\bibitem[{Nielsen et~al.(2017) Nielsen and Critchley and Dodson}]{NieCD17}
Nielsen, F. \& Critchley, F. \& Dodson, C. T. J. (2017). Computational information geometry: \small{For image and signal processing}. (Signals and communication technology).
Springer

\bibitem[{Owadally(2012)}]{Owa12}
Owadally, I. (2012). An Improved Closed Form Solution for the Constrained Minimization of the Root of a Quadratic Functional.
Journal of Computational and Applied Mathematics 236~(17), 4428 -- 4435.

\bibitem[{P\'{e}rez-Cruz(2008)}]{Per08}
P\'{e}rez-Cruz, F. (2008). Kullback-{L}eibler Divergence Estimation of Continuous Distributions.
IEEE International Symposium on Information Theory. (ISIT 2008. IEEE), 1666 -- 1670.

\bibitem[{Schneider and Schweizer(2015)}]{SchS15}
Schneider, J. C. \& Schweizer, N. (2015). Robust measurement of (heavy-tailed) risks: Theory and implimentation.
Journal of Economic Dynamics \& Control 61, 183 -- 203.

\bibitem[{Wang et~al.(2009) Wang and Kulkarni and Verd\'{u}}]{WanKV09}
Wang, Q. \& Kulkarni, S. R. \& Verd\'{u}, S. (2009). Divergence estimation for multidimensional densities via k-nearest-neighbor distances.
IEEE Transactions on Information Theory 55~(5), 2392 -- 2405.
%


\end{thebibliography}
\biboptions{authoryear}

\end{document}